%% file: amsart.tex
\documentclass[11pt,reqno]{amsart}
\pdfoutput=1
\def\keywords{\smallskip\noindent\textsc{Key Words. }}
\def\pbfont#1{\textsf{#1}}

\newif\ifsubmission
\submissionfalse
\usepackage{enumitem}
\usepackage[english]{babel}
\usepackage{mathdots,mathrsfs}
\newcommand{\eatpunct}[1]{\nopunct}
\newcommand{\appref}[1]{Appendix~\ref{#1}}
\usepackage{comment}\excludecomment{vlong}\includecomment{vshort}
\usepackage[numbers]{natbib}\renewcommand{\cite}{\citep}

\makeatletter
\def\NAT@spacechar{~}%
\makeatother
\usepackage{amsfonts,amstext,amsmath,amssymb,stmaryrd,calc,amsthm,mathtools,bold-extra}
\providecommand{\qedhere}{\qed}
\usepackage{thmtools}
\newtheoremstyle%
  {pb}%
  {1pt}%
  {0pt}%
  {\normalfont}%
  {}%
  {\bfseries\boldmath}%
  {}%
  { }%
  {pb}%
\makeatletter
\def\thmhead@pb#1#2#3{%
  {\sffamily\pbname}%
  \thmnote{ {\the\thm@notefont(#3)}}}
\theoremstyle{pb}
\thm@notefont{\sffamily}
\newtheorem{pb}{Problem}
\newenvironment{problem}[2][]
  {\def\pbname{#2}%
    \begin{pb}[\textsf{#1}]%
      \hfill}
  {\end{pb}}
\makeatother
\usepackage{macros}
\theoremstyle{plain}
\newtheorem{theorem}{Theorem}[section]
\newtheorem{lemma}[theorem]{Lemma}
\newtheorem{corollary}[theorem]{Corollary}
\newtheorem{proposition}[theorem]{Proposition}

\theoremstyle{definition}

\theoremstyle{remark}

\newtheorem{remark}[theorem]{Remark}
\usepackage[norelsize]{algorithm2e}

\makeatletter\def\algocf@capseparator{.}\makeatother
\usepackage{subfig,tikz}
\usepackage{url,hyperref}

\providecommand{\urlstyle}[1]{}
\providecommand{\doi}[1]{\href{http://dx.doi.org/#1}{\nolinkurl{doi:#1}}}
\providecommand{\email}[1]{\href{mailto:#1}{\nolinkurl{#1}}}
\begin{document}
\renewcommand{\sectionautorefname}{Section}
\renewcommand{\subsectionautorefname}{Section}
\renewcommand{\subsubsectionautorefname}[1]{\S}
\title{Complexity Hierarchies Beyond Elementary}
\author[S.~Schmitz]{Sylvain Schmitz}
\address{LSV, ENS Cachan \& CNRS \& Inria\\Universit\'e Paris-Saclay\\France}
\email{schmitz@lsv.ens-cachan.fr}
\thanks{Work supported in part by the \textsc{ReacHard} project (ANR
  11 BS02 001 01).}
\begin{abstract}
\input{abstract}
\keywords Fast-growing complexity, subrecursion, well-quasi-order
\end{abstract}
\maketitle

\section{Introduction}
\input{sec-intro}

\section{Fast-Growing Complexity Classes}\label{sec-fg}
\input{sec-fg}

\section{Fast-Growing Complexities in Action}\label{sec-ex}
\input{sec-ex}

\section{Robustness}\label{sec-robust}
\input{sec-robust}

\section{Strictness}\label{sec-strict}
\input{sec-strict}

\input{sec-bestiary}

\section{Concluding Remarks}
\input{sec-concl}

\appendix
\section{Subrecursive Hierarchies}
\label{app-subrec}
\input{appendix}
\bibliographystyle{abbrvnat}
\bibliography{journalsabbr,conferences,fgcc}
\end{document}

%% file: abstract.tex
We introduce a hierarchy of fast-growing complexity classes and show
its suitability for completeness statements of many non elementary
problems.  This hierarchy allows the classification of many decision
problems with a non-elementary complexity, which occur naturally in
logic, combinatorics, formal languages, verification, etc., with
complexities ranging from simple towers of exponentials to
Ackermannian and beyond.

%% file: sec-intro.tex
\textit{Complexity classes,} along with the associated notions of
reductions and completeness, provide our best theoretical tools to
classify and compare computational problems.  The richness and
liveness of this field can be experienced by taking a guided tour of
the \emph{Complexity
  Zoo},\footnote{\url{https://complexityzoo.uwaterloo.ca}.} which
presents succinctly most of the known specimens.  The visitor will find
there a wealth of classes at the frontier between tractability and
intractability, starring the classes P and NP, as they help in
understanding what can be solved efficiently by algorithmic means.

From this tractability point of view, it is not so surprising to find
much less space devoted to the ``truly intractable'' classes, in the
exponential hierarchy and beyond.  Such classes are nevertheless quite
useful for classifying problems, and employed routinely in logic,
combinatorics, formal languages, verification, etc.\ since the 70's
and the exponential lower bounds proven by Meyer and
Stockmeyer~\citep{meyer72,stockmeyer73}.  
\subsubsection*{Non Elementary Problems}
Actually, these two seminal articles go further than mere exponential
lower bounds: they show respectively that satisfiability of the weak
monadic theory of one successor (\pbfont{WS1S}) and equivalence of
star-free expressions (\pbfont{SFEq}) are
\emph{non elementary}, as they require space bounded above and below
by towers of exponentials of height depending (elementarily) on the
size of the input.  Those are just two examples among many others of
problems with non elementary complexities \citep[see
e.g.][]{meyer74,friedman99,vorobyov04}, but they are actually good
representatives of problems with a tower of exponentials as
complexity, i.e., one would expect them to be \emph{complete} for some
suitable complexity class.

What might then come as a surprise is the fact that, presently, the Zoo
does not provide any intermediate stops where classical problems like
\pbfont{WS1S} and \pbfont{SFEq} would fit adequately: they are not in
\textsc{Elementary} (henceforth \elem), but the next class is
\textsc{Primitive-Recursive} (aka \pr), which is far too big:
\pbfont{WS1S} and \pbfont{SFEq} are not \emph{hard} for \pr\ under any
reasonable notion of reduction.  In other words, we seem to be missing
a ``\Tow'' complexity class, which ought to sit somewhere between
\elem\ and \pr.  Going higher, we find a similar uncharted area
between \pr\ and \textsc{Recursive} (aka~R).  These absences are not
specific to the Complexity Zoo: they seem on the contrary universal in
textbooks on complexity theory---which seldom even mention \elem\ or
\pr.  Somewhat oddly, the complexities above R are better explored and
can rely on the arithmetical and analytical hierarchies.

Drawing distinctions based on complexity characterisations can guide
the search for practically relevant restrictions to the problems.  In
addition, non elementary problems are much more pervasive now than in
the 70's, and they are also considered for practical applications,
motivating the implementation of tools, e.g.\ MONA for
\pbfont{WS1S}~\citep{mona}.  It is therefore high time for the
definition of hierarchies suited for their classification.

\subsubsection*{Our Contribution}
In this paper, we propose an ordinal-indexed hierarchy
$(\FC\alpha)_\alpha$ of \emph{fast growing} complexity classes for
non elementary complexities.  Beyond the already mentioned
$\Tow\eqdef\FC3$---for which \pbfont{WS1S} and \pbfont{SFEq} are
examples of complete problems---, this hierarchy includes non
primitive-recursive classes, for which quite a few complete problems
have arisen in the recent years, e.g.
\begin{itemize}
\item $\FC\omega$
in~\citep{fct,urquhart99,phs-mfcs2010,figueira12,bresolin2012,lazic13,hofman14,hague14},
\item $\FC{\omega^\omega}$
in~\citep{lcs,mtl,ata,tsoreach,pepreg,barcelo12,rosavelardo14},
\item $\FC{\omega^{\omega^\omega}}$ in~\citep{HSS-lics2012}, and
\item $\FC{\ezero}$ in~\citep{HaaseSS13,decker15}.
\end{itemize}

The classes $\FC\alpha$ are related to the Grzegorczyk
$(\GH{k})_k$~\citep{grzegorczyk53} and extended Grzegorczyk
$(\FGH{\alpha})_\alpha$~\citep{lob70} hierarchies, which have been
used in complexity statements for non elementary bounds.  The
$(\FGH{\alpha})_\alpha$ classes are very well-suited for
characterising various classes of functions, for instance computed by
forms of \texttt{for} programs~\citep{meyer67} or terminating
\texttt{while} programs~\citep{fairtlough92}, or provably total in
fragments of Peano arithmetic~\citep{fairtlough98,sw12}, and they
characterise some important milestones like \elem\ or \pr.  They are
however too large to classify our decision problems and do not lead to
\emph{completeness} statements---in fact, one can show that there are
\emph{no} ``\elem-complete'' nor ``\pr-complete'' problems---; see
\autoref{sec-fg}.  Our $\FC\alpha$ share however several nice
properties with the $\FGH{\alpha}$ classes: for instance, they form a
strict hierarchy (\autoref{sec-strict}) and are \emph{robust} to
slight changes in their generative functions and to changes in the
underlying model of computation (\autoref{sec-robust})%
.

In order to argue for the suitability of the classes $\FC\alpha$ for
the classification of high-complexity problems, we sketch two
completeness proofs in \autoref{sec-ex}, and present an already long
list of complete problems for $\FC\omega$ and beyond in
\autoref{sec-bestiary}.  A general rule of thumb seems to be that
statements of the form ``$L$ is in $\FGH{\alpha}$ but not in
$\FGH{\beta}$ for any $\beta<\alpha$'' found in the literature can
often be replaced by the much more precise ``$L$ is
$\F\alpha$-complete.''

There are of course essential limitations to our approach: there is no
hope of defining such ordinal-indexed hierarchies that would exhaust R
using sensible ordinal notations~\citep{feferman62}; this is called
the ``subrecursive stumbling block''
by \citet[\sectionautorefname~5.1]{sw12}.  Our aim here is more
modestly to provide suitable definitions ``from below'' for
naturally-occurring complexity classes above \elem.

In an attempt not to drown the reader in the details of subrecursive
functions and their properties, most of the technical contents appears
in \appref{app-subrec} at the end of the paper.

%% file: sec-fg.tex
We define in this section the complexity classes $\F\alpha$.  We rely
for this on the fast-growing functions $F_\alpha$ of \citet{lob70} as
a standard against which we can measure high complexities
(c.f.~\autoref{ssub-fg}).  In logic and recursion theory, these
functions are used to generate the classes of functions $\FGH\alpha$
when closed under substitution and limited primitive recursion
(see~\autoref{ssub-rec}).  These classes are however not suitable for
our complexity classification objectives: the class $\FGH\alpha$
contains indeed arbitrary finite compositions of the function
$F_\alpha$.  We define instead in \autoref{sub-fgcc} each $\F\alpha$
class as the class of problems decidable within time bounded by a
single application of $F_\alpha$ composed with any function $p$
already defined in the lower levels $\FGH\beta$ for $\beta<\alpha$.

These hierarchies of functions, function classes, and complexity
classes we employ in order to deal with non elementary complexities
are all indexed using ordinals, and we reuse the very rich literature
on subrecursion~\citep[e.g.][]{rose84,odifreddi99,sw12}.  We strive to
employ notations compatible with those of \citet[Chapter~4]{sw12}, and
refer the interested reader to their monograph for proofs and
additional material.

\subsection{Cantor Normal Forms and Fundamental Sequences}
In this paper, we only deal with ordinals that can be denoted
syntactically as terms in \emph{Cantor Normal Form}:
\begin{align}\label{CNF}
  \alpha = \omega^{\alpha_1}\cdot c_1+\cdots
  +\omega^{\alpha_n}\cdot c_n&\text{ where }\alpha > \alpha_1
  >\cdots> \alpha_n\text{ and }\omega>c_1,\ldots,c_n>0\tag{CNF}
\end{align}
and hereditarily $\alpha_1,\dots,\alpha_n$ are also written
in \ref{CNF}.  In this representation, $\alpha=0$ if and only if
$n=0$.  An ordinal $\alpha$ with \ref{CNF} of form $\alpha'+1$ is
called a
\emph{successor} ordinal---it has $n>0$ and $\alpha_n=0$---, and
otherwise if $\alpha>0$ it is called a \emph{limit} ordinal, and can
be written as $\gamma+\omega^\beta$ by setting
$\gamma=\omega^{\alpha_1}\cdot c_1+\cdots
+\omega^{\alpha_n}\cdot(c_n-1)$ and $\beta=\alpha_n$.  We usually employ
``$\lambda$'' to denote limit ordinals.

A \emph{fundamental sequence} for a limit ordinal $\lambda$ is a
sequence $(\lambda(x))_{x<\omega}$ of ordinals with supremum
$\lambda$.  We consider a standard assignment of fundamental sequences
for limit ordinals, which is defined inductively by
\begin{align}
  \label{eq-fund-def}
  (\gamma+\omega^{\beta+1})(x)&\eqdef\gamma+\omega^\beta\cdot (x+1)\;,&
  (\gamma+\omega^{\lambda})(x)&\eqdef\gamma+\omega^{\lambda(x)}\;.
\end{align}
This particular assignment of fundamental sequences satisfies e.g.\
$0<\lambda(x)<\lambda(y)$ for all $x<y$ and limit ordinals $\lambda$.
For instance, $\omega(x)=x+1$,
$(\omega^{\omega^4}+\omega^{\omega^3+\omega^2})(x)=\omega^{\omega^4}+\omega^{\omega^3+\omega\cdot(x+1)}$.
We also consider the ordinal $\ezero$, which is the supremum of all
the ordinals writable in \ref{CNF}, as a limit ordinal with
fundamental sequence defined by $\ezero(0)\eqdef\omega$ and
$\ezero(x+1)\eqdef\omega^{\ezero(x)}$, i.e.\ a tower of $\omega$'s of
height $x+1$.

\subsection{The Extended Grzegorczyk Hierarchy}
\ifsubmission This is an ordinal-indexed infinite hierarchy of classes
$(\FGH{\alpha})_{\alpha<\ezero}$ \else
$(\FGH{\alpha})_{\alpha<\ezero}$ is an ordinal-indexed infinite
hierarchy of classes \fi of functions with argument(s) and images in
$\+N$ \citep{lob70}.  The extended Grzegorczyk hierarchy has multiple
natural characterisations: for instance as \texttt{loop} programs for
$\alpha<\omega$ \citep{meyer67}, as ordinal-recursive functions with
bounded growth \citep{wainer70}, as functions computable with
restricted resources as we will see in \eqref{eq-fgh-char}, as
functions that can be proven total in fragments of Peano arithmetic
\citep{fairtlough98}, etc.

\subsubsection{Fast-Growing Functions}\label{ssub-fg}
At the heart of each $\FGH{\alpha}$ lies the
$\alpha$th \defstyle{fast-growing function}
\mbox{$F_\alpha{:}\,\+N\to\+N$}, which is defined inductively on the
ordinal index: as the successor function at index $0$
\begin{align}\label{eq-f-zero}
  F_0(x)&\eqdef x+1\:,
  \shortintertext{by iteration at successor indices $\alpha+1$}
  \notag\\[-2em]
  \label{eq-f-succ}
  F_{\alpha+1}(x)&\eqdef F_\alpha^{\omega(x)}(x)=\overbrace{F_\alpha(\cdots
    (F_\alpha}^{\omega(x)\text{ times}}(x))\cdots)\;,
  \shortintertext{and by diagonalisation on the fundamental sequence at
    limit indices $\lambda$}
    \label{eq-f-lim}
  F_\lambda(x)&\eqdef F_{\lambda(x)}(x)\:.
\end{align}
For instance, $F_1(x)=2x+1$, $F_2(x)=2^{x+1}(x+1)-1$, $F_3$ is a non
elementary function that grows faster than
$\tow(x)\eqdef\tetration{2}{x}$, $F_\omega$ a non primitive-recursive
``Ackermannian'' function, $F_{\omega^\omega}$ a non
multiply-recursive ``hyper-Ackermannian'' function, and
$F_{\ezero}(x)$ cannot be proven total in Peano arithmetic.  For every
$\alpha$, the $F_\alpha$ function is strictly \emph{monotone} in its
argument, i.e.\ $x<y$ implies $F_\alpha(x)<F_\alpha(y)$.  As
$F_\alpha(0)=1$, it is therefore also strictly \emph{expansive}, i.e.\
$F_\alpha(x)>x$ for all $x$.

\subsubsection{Computational Characterisation}
The extended Grzegorczyk hierarchy itself is defined by means of
recursion schemes with the $(F_\alpha)_\alpha$ as generators
(see \autoref{ssub-rec}).  Nevertheless, for $\alpha\geq 2$, each of
its levels $\FGH{\alpha}$ is also characterised as a class of
functions computable with bounded resources~\citep{wainer70}.  More
precisely, for $\alpha\geq 2$, it is the class of functions computable
by deterministic Turing machines in time bounded by $O(F^c_\alpha(n))$
for some constant $c$, when given an input of size $n$:
\begin{align}\label{eq-fgh-char}
  \FGH{\alpha}&=\bigcup_{c<\omega}\CC{FDTime}\left(F_\alpha^c(n)\right).
\end{align}
Note that the choice between deterministic and nondeterministic, or
between time-bounded and space-bounded computations in
\eqref{eq-fgh-char} is irrelevant, because \mbox{$\alpha\geq 2$} and
$F_2$ is already a function of exponential growth.

\subsubsection{Main Properties}
Each class $\FGH{\alpha}$ is closed under (finite) composition.  Every
function $f$ in $\FGH{\alpha}$ is \defstyle{honest}, i.e.\ can be
computed in time bounded by some function also in
$\FGH\alpha$~\citep{wainer70,fairtlough98}---this is a relaxation of
the \defstyle{time constructible} condition, which asks instead for
computability in time $O(f(n))$.  Since each $f$ in $\FGH{\alpha}$ is
also bounded by $F_\alpha^c$ for some $c$
\citep[\theoremautorefname~2.10]{lob70}, this means that
\begin{equation}\label{eq-honest}
  \FGH\alpha=\bigcup_{f\in\FGH\alpha}\CC{FDTime}\left(f(n)\right)\,.
\end{equation}
In particular, for every $\alpha$ the function $F_\alpha$ belongs to $\FGH{\alpha}$, and therefore $F^c_\alpha$ also belongs to
$\FGH{\alpha}$.

Every $f$ in $\FGH{\beta}$ is also eventually bounded by
$F_{\alpha}$ if $\beta<\alpha$~\citep{lob70}, i.e.\ there exists a
rank $x_0$ such that, for all $x_1,\dots,x_n$, if $\max_i x_i\geq
x_0$, then $f(x_1,\dots,x_n)\leq F_{\alpha}(\max_i x_i)$---a fact
that we will use copiously.  However, for all $\alpha>\beta>0$,
$F_{\alpha}\not\in\FGH{\beta}$, and the hierarchy
$(\FGH{\alpha})_{\alpha<\ezero}$ is therefore strict for $\alpha>0$.

\subsubsection{Milestones}
At the lower levels, $\FGH{0}=\FGH{1}$ contains (among others) all the
linear functions (see \autoref{ssub-lin}).  We focus however in this
paper on the non elementary classes by restricting ourselves to
$\alpha\geq 2$.  Writing %
\begin{equation}\label{eq-fgh-less}
  \FGH{<\alpha}\eqdef\bigcup_{\beta<\alpha}\FGH\beta\;,
\end{equation}
we find for instance
$\FGH{2}=\FGH{<3}=\CC{F\elem}$ the set of Kalmar-elementary functions,
$\FGH{<\omega}=\CC{F\pr}$ the set of primitive-recursive functions,
$\FGH{<\omega^\omega}=\CC{FMR}$ the set of multiply-recursive
functions, and $\FGH{<\ezero}=\CC{FOR}$ the set of ordinal-recursive
functions (up to $\ezero$).  We are dealing here with classes of
functions, but writing $\FGH{\alpha}^\ast$ for the restriction of
$\FGH{\alpha}$ to $\{0,1\}$-valued functions, i.e.
\begin{align}\label{eq-fgh-dec}
  \FGH{\alpha}^\ast&=\bigcup_{c<\omega}\CC{DTime}\left(F_\alpha^c(n)\right),&
  \FGH{<\alpha}^\ast&\eqdef\bigcup_{\beta<\alpha}\FGH\beta^\ast\;,
  \end{align}
we obtain the corresponding classes for decision problems
$\FGH{<3}^\ast=\elem$, $\FGH{<\omega}^\ast=\pr$,
$\FGH{<\omega^\omega}^\ast=\CC{MR}$, and
$\FGH{<\ezero}^\ast=\CC{OR}$.

\subsection{Fast-Growing Complexity Classes.}\label{sub-fgcc}
Unfortunately, the classes in the extended Grzegorczyk hierarchy are
not quite satisfying for some interesting problems, which are non
elementary (or non primitive-recursive, or non multiply-recursive,
\dots), but only \emph{barely} so.  The issue is that complexity
classes like e.g.\ $\FGH{3}^\ast$, which is the first class to contain
non elementary problems, are very large: $\FGH{3}^\ast$ contains for
instance problems that require space $F_{3}^{100}(n)$, more than a
hundred-fold compositions of towers of exponentials.  As a result,
hardness for $\FGH{3}^\ast$ cannot be obtained for many classical examples
of non elementary problems.

We therefore introduce smaller classes of %
problems:
\begin{align}\label{eq-fast-class}
 \F\alpha&\eqdef\bigcup_{p\in\FGH{<\alpha}}\CC{DTime}\left(F_\alpha(p(n))\right)\;.
\end{align}
In contrast with $\FGH{\alpha}^\ast$ in~\eqref{eq-fgh-dec}, only a single
application of $F_\alpha$ is possible, composed with some ``lower''
\emph{reduction} function $p$ from $\FGH{<\alpha}$.  As previously,
the choice of \CC{DTime} rather than \CC{NTime} or \CC{Space} is
irrelevant for $\alpha\geq 3$ (see \autoref{th-redc} later).

This definition yields for instance the desired class
$\Tow\eqdef\F3$,
closed under elementary reductions (i.e., reductions
    in $\FGH{2}$), but also a class %
$\Ack\eqdef\F\omega$
of Ackermannian problems closed under
    primitive-recursive reductions, a class %
$\hack\eqdef\F{\omega^\omega}$
of hyper-Ackermannian problems closed under multiply-recursive
reductions, etc.  In each case, we can think of $\F\alpha$ as the
class of problems not solvable with resources in $\FGH{<\alpha}$, but
barely so: non elementary problems for $\F3$, non primitive-recursive
ones for $\F\omega$, non multiply-recursive ones for
$\F{\omega^\omega}$, and so on.  See \autoref{fig-fg} for the first
main stops of the hierarchy.

\begin{figure}[t]
  \centering
  \begin{tikzpicture}[every node/.style={font=\small}]
    \draw[color=blue!90,fill=blue!20,thick](-1.1,0) arc (180:0:5.4cm);
    \shadedraw[color=black!90,top color=black!20,middle color=black!5,opacity=20,shading angle=-15](-1,0) arc (180:0:5cm);
    \draw[color=blue!90,fill=blue!20,thick](-.6,0) arc (180:0:3.3cm);
    \shadedraw[color=black!90,top color=black!20,middle color=black!5,opacity=20,shading angle=-15](-.5,0) arc (180:0:3cm);
    \draw[color=blue!90,fill=blue!20,thick](-.1,0) arc (180:0:1.7cm);
    \shadedraw[color=black!90,top color=black!20,middle
    color=black!5,opacity=20,shading angle=-15,thin](0,0) arc (180:0:1.5cm);
    \draw (1.5,.5) node{$\FGH{<3}^\ast=\elem$};
    \draw (4,1.2) node[color=blue]{$\F3=\Tow$};
    \draw[color=blue!90,thick] (3.15,1) -- (3.05,.9);
    \draw (2.5,2) node{$\FGH{<\omega}^\ast=\CC{PR}$};
    \draw (6.3,2) node[color=blue]{$\F\omega=\Ack$};
    \draw[color=blue!90,thick] (5.65,1.8) -- (5.55,1.7);
    \draw (4,3.8) node{$\FGH{<\omega^\omega}^\ast=\CC{MR}$};
    \draw (10,3) node[color=blue]{$\F{\omega^\omega}=\hack$};
    \draw[color=blue!90,thick] (9.1,2.8) -- (9,2.7);
    \draw (8.5,5) node[rotate=40,font=\Large]{$\cdots$};
  \end{tikzpicture}
  \caption{Some complexity classes beyond \elem.\label{fig-fg}}
\end{figure}
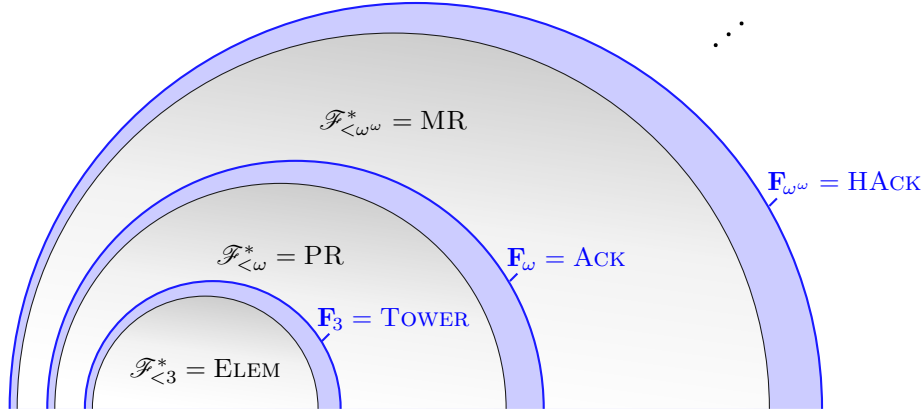

\subsubsection{Reduction Classes}\label{sss-red}
Of course, we could replace in \eqref{eq-fast-class} the class of
reductions $\FGH{<\alpha}$ by a more traditional one, like logarithmic
space (\CC{FL}) or polynomial time (\CC{FP}) functions.  We feel
however that our definition in \eqref{eq-fast-class} better captures
the intuition we have of a problem being ``complete for $F_\alpha$.''
Moreover, using at least $\FGH2$ as our class of reductions allows to
effectively compute the $F_\alpha$ function in the functional version
$\CC{F}\F\alpha$ of $\F\alpha$ (see \autoref{sub-constr}), leading to
interesting combinatorial algorithms (see \autoref{ssub-algo} for an
example).

Unless stated differently, we always assume many-one $\FGH{<\alpha}$
reductions when discussing hardness for $\FC\alpha$ in the remainder
of this paper, but we could just as easily consider Turing reductions
(see \autoref{ssub-red}).

\subsubsection{Basic $\FC\alpha$-Complete Problems}\label{sss-compl}
By \eqref{eq-fast-class}, $\FC\alpha$-hardness proofs can reduce from
the acceptance problem of some input string $x$ by some deterministic
Turing machine $M$ working in time $F_\alpha(p(n))$ for some $p$ in
$\FGH{<\alpha}$.  This can be simplified to a machine $M'$ working in
time $F_\alpha(n)$.  Indeed, because $p$ in $\FGH{<\alpha}$ is honest,
$p(n)$ can be computed in $\FGH{<\alpha}$.  Thus the acceptance of $x$
by $M$ can be reduced to the acceptance problem of a $\#$-padded input
string $x'\eqdef x\#^{p(|x|)-|x|}$ of length $p(|x|)$ by a machine
$M'$ that simulates $M$, and treats $\#$ as a blank symbol---now $M'$
works in time $F_\alpha(n)$.  Another similarly basic $\FC\alpha$-hard
problem is the halting problem for Minsky machines with the sum of
counters bounded by $F_\alpha(n)$~\citep[see][]{fischer68}.

To sum up, we have by definition of the $(\F\alpha)_\alpha$ classes
the following two $\F\alpha$-complete problems---which incidentally
have been used in most of the master reductions in the literature in
order to prove non primitive-recursiveness, non
multiple-recursiveness, and other hardness results~\citep{jancar,urquhart99,phs-mfcs2010,lcs,HSS-lics2012,HaaseSS13,lazic13,rosavelardo14,decker15}:

\begin{problem}[$\mathsf F_\alpha$-TM]{$F_\alpha$-Bounded Turing
    Machine Acceptance}\label{pb-tm}\hfill%
\begin{description}[topsep=0pt,itemsep=0pt,partopsep=0pt,parsep=0pt]
\item[instance:] A deterministic Turing machine $M$ working in time
  $F_\alpha$ and an input $x$.
\item[question:] Does $M$ accept $x$?
\end{description}
\end{problem}

\begin{problem}[$\mathsf F_\alpha$-MM]{$F_\alpha$-Bounded Minsky Machine
    Halting}\label{pb-mm}\hfill%
\begin{description}[topsep=0pt,itemsep=0pt,partopsep=0pt,parsep=0pt]
\item[instance:] A deterministic Minsky machine $M$ with sum of
  counters bounded by $F_\alpha(|M|)$.
\item[question:] Does $M$ halt?
\end{description}
\end{problem}\noindent
See \autoref{sec-bestiary} for a catalogue of natural complete problems,
which should be easier to employ in reductions.

%% file: sec-ex.tex
We present now two short tutorials for the use of fast-growing
complexities, namely for the equivalence problem for start-free
expressions (\autoref{sub-ws1s}) and reachability in lossy counter
systems (\autoref{sub-lcm}), pointing to the relevant technical
results from later sections.  We also briefly discuss in each case the
palliatives employed so far in the literature for expressing such
complexities.

\subsection{A $\CC{Tower}$-Complete Example\eatpunct.}\label{sub-ws1s}
\ifsubmission Such an example \fi can be found in the seminal
paper of \citet{stockmeyer73}, and is quite likely already known by
many readers.  Define a \defstyle{star-free expression} over
some alphabet $\Sigma$ as a term $e$ with abstract syntax
\begin{equation*}
  e ::= a \mid \varepsilon\mid\emptyset \mid e+e\mid ee\mid
  \neg e
\end{equation*}
where ``$a$'' ranges over $\Sigma$ and ``$\varepsilon$'' denotes the empty
string.  Such expressions are inductively interpreted as languages
included in $\Sigma^\ast$ by:
\begin{align*}
  \sem{a}&\eqdef\{a\} &
  \sem{\varepsilon}&\eqdef\{\varepsilon\}&
  \sem{\emptyset}&\eqdef\emptyset\\
  \sem{e_1+e_2}&\eqdef\sem{e_1}\cup\sem{e_2}&
  \sem{e_1e_2}&\eqdef\sem{e_1}\cdot\sem{e_2}&
  \sem{\neg e}&\eqdef\Sigma^\ast\setminus\sem{e}\;.
\end{align*}
The decision problem \pbfont{SFEq} asks, given two such expressions
$e_1,e_2$, whether they are \emph{equivalent}, i.e.\ whether
$\sem{e_1}=\sem{e_2}$.  \Citet{stockmeyer73} show that this problem is
hard for $\tow(\log n)$ space under \CC{FL} reductions if
$|\Sigma|\geq 2$.  The problem \pbfont{WS1S} can be shown similarly
hard thanks to a reduction from \pbfont{SFEq}.

\subsubsection{Completeness}  Recall that \Tow\ is defined as $\F3$,
i.e.\ by the instantiation of \eqref{eq-fast-class} for $\alpha=3$, as
the problems decidable by a Turing machine working in time $F_3$ of
some elementary function of the input size:
\begin{equation}
  \label{eq-tower-class}
  \Tow\eqdef\F3=\!\!\bigcup_{p\in\CC F\elem}\!\!\CC{DTime}\left(F_3(p(n))\right)\;.
\end{equation}

Once hardness for $\tow(\log n)$ is established, hardness for $\Tow$
under elementary reductions is immediate; a detailed proof can apply
\autoref{th-ack} and Equation~\eqref{eq-fast-space} to show that
\begin{equation}
\Tow=\bigcup_{p\in\CC F\elem}\CC{Space}(\tow(p(n))
\end{equation} and use a padding
argument as in \autoref{sss-compl} to conclude.  

That \pbfont{SFEq} is in $\Tow$ can be checked using an
automaton-based algorithm: construct automata recognising $\sem{e_1}$
and $\sem{e_2}$ respectively, using determinization to handle each
complement operator at the expense of an exponential blowup, and check
equivalence of the obtained automata in \CC{PSpace}---the overall
procedure is in space polynomial in $\tow(n)$, thus in $\F3$.  A
similar automata-based procedure yields the upper bound for
\pbfont{WS1S}.

\subsubsection{Discussion}
Regarding upper bounds, there was a natural candidate in the
literature for the missing class \Tow: \citet{grzegorczyk53} defines
an infinite hierarchy of function classes $(\GH{k})_{k\in\+N}$ inside
F\pr\ with $\GH{k+1}=\FGH{k}$ for $k\geq 2$.  This yields
$\text{F\elem}=\GH{3}$, and the $\tow$ function is in
$\GH{4}\setminus\GH{3}$.  Thus \pbfont{WS1S} and \pbfont{SFEq} are in
``time $\GH{4}$,'' and such a notation has occasionally been employed,
for instance for \pbfont{$\beta$-Eq} the $\beta$ equivalence of simply
typed $\lambda$-terms~\citep{statman79,schwichtenberg82,beckmann01}.
Again, we face the issue that $\GH{4}$ is much too large a resource
bound, as it contains for instance all the finite iterates of the
$\tow$ function, and there is therefore no hope of proving the
hardness for $\GH{4}$ of \pbfont{WS1S}, \pbfont{SFEq}, or indeed
\pbfont{$\beta$-Eq}, at least if using a meaningful class of
reductions.

Regarding non elementary lower bounds, recent papers typically
establish hardness for $k$-\CC{ExpTime} (or $k$-\CC{ExpSpace})
\emph{for infinitely many $k$} (possibly through a suitable
parametrisation of the problem at hand), for instance by reducing
from the acceptance of an input of size $n$ by a $\tetra[2]{\!n}{k}$
time-bounded Turing machine.  Provided that such a lower bound
argument is \emph{uniform} for those infinitely many $k$, it
immediately yields a \Tow-hardness proof, by choosing $k\geq n$.  On a
related topic, note that, in contrast with e.g.\ the relationship
between \CC{PH} and \CC{PSpace}, because the exponential hierarchy is
known to be strict, we know for certain
that%
\begin{itemize}
\item for all $k$, $k\text{-}\CC{ExpTime}\subsetneq\elem=\bigcup_kk\text{-}\CC{ExpTime}$,
\item there are no ``\CC{Elem}-complete problems,'' and
\item $\CC{Elem}\subsetneq\Tow$.
\end{itemize}

\subsection{An \Ack-Complete Example}\label{sub-lcm}
Possibly the most popular complete problem for \Ack\ in use in
reductions, \probref{pb-lcn} \pbfont{Reachability} asks whether a
given configuration is reachable in a \emph{lossy counter machine}
(LCM)~\citep{phs-mfcs2010}.  Such counter machines are syntactically
identical to Minsky machines $\tup{Q,\mathtt C,\delta,q_0}$, where
transitions $\delta\subseteq Q\times
\mathtt C\times\{\zero,\incr,\decr\}\times Q$ operate on a set
$\mathtt C$ of counters through \emph{zero-tests} $\mathtt
c\zero$, \emph{increments} $\mathtt c\incr$ and \emph{decrements}
$\mathtt c\decr$.  The semantics of an LCM differ however from the
usual, ``reliable'' semantics of a counter machine in that the counter
values can decrease in an uncontrolled manner at any point of the
execution.  These unreliable behaviours make several problems decidable
on LCMs, contrasting with the situation with Minsky machines.

Formally, a configuration $\sigma=(q,\vec v)$ associates a control
location $q$ in $Q$ with a counter valuation $\vec v$ in $\+N^{\mathtt
C}$, i.e.\ counter values can never go negative.  A transition of the
form $(q,\mathtt c,\mathtt{op},q')$ defines a computation step $(q,\vec
v)\to(q',\vec v')$ where $\vec v(\mathtt c')\leq\vec v'(\mathtt c')$
for all $\mathtt c\neq\mathtt c'$ in $\mathtt C$, and
\begin{itemize}
\item if $\mathtt{op}=\zero$, then $\vec v(\mathtt c)\geq\vec v'(\mathtt
c)=0$,
\item if $\mathtt{op}=\incr$, then $\vec v(\mathtt c)+1\geq\vec
v'(\mathtt c)$, and
\item if $\mathtt{op}=\decr$, then $\vec v(\mathtt c)\geq\vec
v'(\mathtt c)+1$.
\end{itemize}

Let the initial configuration be $(q_0,\vec 0)$.  The reachability
problem for such a system asks whether a given configuration $\tau$
can be reached in a finite number of steps, i.e.\ whether $(q_0,\vec
0)\to^\ast \tau$.  The hardness proof of \citet{phs-mfcs2010}
immediately yields that this problem is \Ack-hard \citep[see
also][]{urquhart99,phs-IPL2002}, where \Ack\ is defined as an instance
of \eqref{eq-fast-class}: it is the class of problems decidable with
$F_\omega$ resources of some primitive-recursive function of the input
size:
\begin{equation}
  \Ack\eqdef \F\omega=\bigcup_{p\in\CC
    F\pr}\CC{DTime}\big(F_\omega(p(n))\big)\;.
\end{equation}

\subsubsection{Decidability of \probref{pb-lcn}} Lossy counter
machines define \emph{well-structured} transition systems over the set
of configurations $Q\times\+N^{\mathtt C}$, for which generic algorithms
have been designed~\citep{abdulla2000c,finkel98b}, which rely on the
existence of a well-quasi-ordering \citep[wqo, see][]{kruskal72} over
the set of configurations.  The particular variant of the algorithm we
present here is well-suited for a complexity analysis, and is taken
from \citep{concur/SchmitzS13}.

Call a sequence of configurations $\sigma_0,\sigma_{1},\dots,\sigma_n$
a \emph{witness} if $\sigma_0=\tau$ is the target configuration,
$\sigma_n=(q_0,\vec 0)$ is the initial configuration, and
$\sigma_{i+1}\to\sigma_{i}$ for all $0\leq i<n$.  An instance of
\probref{pb-lcn} is positive if and only if there exists a witness,
which we will search for backwards, starting from $\tau$ and
attempting to reach the initial configuration $(q_0,\vec 0)$.

Consider the ordering over configurations defined by $(q,\vec v)\leq
(q',\vec v')$ if and only if $q=q'$ and $\vec v\leq_\times\vec v'$,
the latter being defined as $\vec v(\mathtt c)\leq \vec
v'(\mathtt c)$ for all $\mathtt c$ in $\mathtt C$.  Observe that, if
$\sigma_0,\sigma_{1},\dots,\sigma_n$ is a
\emph{shortest} witness, then for all $i<j$,
$\sigma_i\not\leq\sigma_j$, i.e.\ it is a \emph{bad} sequence for $\leq$,
or we could have picked $\sigma_j$ at step $i$ and obtained a strictly
shorter witness.  Furthermore, if at some step $i$ there existed
$s'_i\leq s_i$ with $s'_i\to s_{i-1}$, then we could substitute $s'_i$
for $s_i$ and still have a witness, because $s_{i+1}\to s'_i$.  Thus,
if there exists a witness, then there is a \emph{minimal bad} one,
i.e.\ a bad one where for all $0<i<n$,
$\sigma_{i+1}\in\MinPre(\sigma_i)$ where $\MinPre(\sigma)\eqdef
\min_\leq\{\sigma'\mid \sigma'\to\sigma\}$.

{Now, because $Q$ and ${\mathtt C}$ are finite, $(Q\times\+N^{\mathtt
C},{\leq})$ is a well-quasi-order by Dickson's Lemma, thus
\renewcommand{\theenumi}{\roman{enumi}}
\renewcommand{\labelenumi}{(\theenumi)}
\begin{enumerate}
\item\label{wqo1} for all $i$, the set $\MinPre(\sigma_i)$ is finite,
  and
\item\label{wqo2} any bad sequence, i.e. any sequence
  $\sigma_0,\sigma_{1},\dots$ where $\sigma_i\not\leq\sigma_j$ for all
  $i<j$, is finite.
\end{enumerate}
Therefore, an algorithm for \probref{pb-lcn} can proceed by exploring
a tree of prefixes of potential minimal witnesses, which has
finite degree by~\eqref{wqo1} and finite height by~\eqref{wqo2}, hence
by K\H{o}nig's Lemma is finite.}

\subsubsection{Length Function Theorems}\label{ssub-length}
A nondeterministic version of this search for a witness for
\probref{pb-lcn} will see its complexity depend essentially on the
height of the tree, i.e.\ on the length of bad sequences.  Define the
size of a configuration as its infinity norm $|(q,\vec
v)|=\max_{\mathtt c\in\mathtt C}\vec v(\mathtt c)$, and note that any
$\sigma$ in $\MinPre(\sigma_{i})$ is of size
$|\sigma|\leq |\sigma_i|+1$.  This means that in any sequence
$\sigma_0,\sigma_1,\dots$ where $\tau=\sigma_0$ and
$\sigma_{i+1}\in\MinPre(\sigma_i)$ for all $i$,
$|\sigma_i|\leq |\tau|+i=\suc^i(|\tau|)$ the $i$th iterate of the
successor function $\suc(x)\eqdef x+1$.  We call such a sequence
\emph{controlled} by $\suc$.

What a \emph{length function theorem} provides is an upper bound on
the length of controlled bad sequences over a wqo, depending on the
control function---here the successor function---and the maximal
order type of the wqo---here $\omega^{|\mathtt C|}\cdot |Q|$.  In our
case, the theorems in \citep{SS-icalp2011,wqo} provide an
\begin{equation}\label{eq-length}
  F_{h,|\mathtt C|}^{|Q|}(|\tau|)\leq
  F_{h,\omega}(\max\{|\mathtt C|,|Q|,|\tau|\})\eqdef \ell
\end{equation}
upper bound on both this length and the maximal size of any
configuration in the sequence, where
\begin{itemize}
\item $h{:}\,\+N\to\+N$ is an increasing polynomial function (which
  depends on the control function) and
\item for any increasing $h{:}\,\+N\to\+N$, $(F_{h,\alpha})_\alpha$ is
  a \emph{relativized} fast-growing hierarchy that uses $h$ instead of
  the successor function as base function with index $0$:
\begin{align}\label{eq-rfg-def}
  F_{h,0}(x)&\eqdef h(x)\;,&
  F_{h,\alpha+1}(x)&\eqdef F_{h,\alpha}^{\omega(x)}(x)\;,&
  F_{h,\lambda}(x)&\eqdef F_{h,\lambda(x)}(x)\;.
\end{align}
\end{itemize}

\subsubsection{A Combinatorial Algorithm}\label{ssub-algo}
We have established an upper bound on the length of a shortest
minimal witness, entailing that if a witness exists, then it is of
length bounded by $\ell$ defined in~\eqref{eq-length}.  This bound can
be exploited by a nondeterministic \emph{forward} algorithm, which
\begin{enumerate}
\item computes $\ell$ in a first phase: as we will see
  with \autoref{th-constr}, this can be performed in time
  $F_{h,\omega}(e(n)))$ for some elementary function~$e$,
\item then nondeterministically explores the reachable configurations,
  starting from the initial configuration $(q_0,\vec 0)$ and
  attempting to reach the target configuration $\tau$---but aborts if
  the upper bound on the length is reached.  This second phase uses at
  most $\ell$ steps, and each step can be performed in time polynomial
  in the size of the current configuration, itself bounded by $\ell$.
  The whole phase can thus be performed in time polynomial in $\ell$,
  which is bounded by $F_{h,\omega}(f(n))$ for some
  primitive-recursive $f$ by \autoref{th-redc}.
\end{enumerate}
Thus the overall complexity of this algorithm can be bounded by
$F_{h,\omega}(p(n))$ where $h$ and $p$ are primitive-recursive.
Because by \autoref{cor-rel} and Equation~\eqref{eq-fast-space}, for any
primitive-recursive strictly increasing $h$,
\begin{equation}
  \Ack = \bigcup_{p\in\CC F\pr}\CC{NTime}\big(F_{h,\omega}(p(n))\big)\;,
\end{equation}
this means that \probref{pb-lcn} is in \Ack.

\subsubsection{Discussion}
The oldest statement of \Ack-completeness (under polynomial time
Turing reductions) we are aware of is due to \citet{clote} for
\probref{pb-fcp}, the finite containment problem for Petri nets; see
\autoref{ssub-vas}.  As observed by \citeauthor{clote}, his definition
of \Ack\ as $\CC{DTime}\big(F_\omega(n)\big)$ is somewhat problematic,
since the class is not robust under changes in the model of
computation, for instance RAM vs.\ multitape Turing machines.  A
similar issue arises with the definition
$\bigcup_{c<\omega}\CC{DTime}\big(F_\omega(n+c)\big)$ employed in
\citep{HSS-lics2012}: though robust under changes in the model of
computation, it is not closed under reductions.  Those classes
are too tight to be convenient.

Conversely, stating that a problem is ``in $\FGH{\omega}^\ast$ but not
in $\FGH{k}^\ast$ for any $k$''~\citep[e.g.][]{FFSS-lics2011} is much
less informative than stating that it is $\F\omega$-complete:
$\FGH\omega^\ast$ is too large to allow for completeness statements,
see \autoref{sec-strict}.

%% file: sec-robust.tex
In the applications of fast-growing classes we discussed in
sections~\ref{sub-ws1s} and~\ref{sub-lcm}, we relied on both counts on
their ``robustness'' to minor changes in their definition.  More
precisely, we employed space or time hierarchies indifferently, and
alternative generative functions: first for the lower bound of
\pbfont{SFEq} and \pbfont{WS1S}, when we used the \tow\ function
instead of $F_3$ in the reduction, and later for the upper bound of
\probref{pb-lcn}, where we relied on a relativised version of
$F_\omega$.  In this section, we prove these and other small changes
to be innocuous.

\subsection{Generative Functions}\label{sub-robust}
There are many variants for the definition of the fast-growing
functions $(F_\alpha)_\alpha$, but they are all known to generate
essentially the same hierarchy $(\FGH\alpha)_\alpha$.\footnote{See
\citep{ritchie65} and \citep[pp.~48--51]{lob70} for such results---and
the works of Weiermann et al.\ on phase transitions for investigations
of when changes \emph{do} have an impact~\citep[e.g.][]{omri09}.}  
Nevertheless, because the fast-growing complexity classes $\F\alpha$
we defined are smaller, there is no guarantee for these classical
results to hold for them.

\subsubsection{Ackermann Hierarchy}\label{sub-ack}
We start with one particular variant, which is rather common in
the literature: define $A_\alpha{:}\,\+N\to\+N$ for $\alpha>0$ by:
\begin{align}\label{eq-ack-def}
  A_1(x) &\eqdef 2x\;,&
  A_{\alpha+1}(x) &\eqdef A_\alpha^x(1)\;,&
  A_\lambda(x) &\eqdef A_{\lambda(x)}(x)\;.
\end{align}
The hierarchy differs in the treatment of successor indices, where the
argument is reset to $1$ instead of keeping $x$ as
in \eqref{eq-f-succ}.  This definition results for instance in
$A_2(x)=2^x$ and $A_3(x)=\tow(x)$, and is typically used in lower
bound proofs.

We can define a hierarchy of decision problems generated from the
$(A_\alpha)_\alpha$ by analogy with \eqref{eq-fast-class}:
\begin{equation}
    \mathbf{A}_\alpha\eqdef\bigcup_{p\in\FGH{<\alpha}}\!\!\CC{DTime}\left(A_\alpha(p(n)\right).
\end{equation}
For two functions $g{:}\,\+N\to\+N$ and $h{:}\,\+N\to\+N$, let us
write $g\leq h$ if $g(x)\leq h(x)$ for all $x$ in $\+N$.  Because
$A_\alpha\leq F_\alpha$ for all $\alpha>0$, it follows that
$\mathbf{A}_\alpha\subseteq\F\alpha$.  The converse inclusion also
holds: in order to prove it, it suffices to exhibit for all $\alpha>0$
a function $p_\alpha$ in $\FGH{<\alpha}$ such that $F_\alpha\leq
A_\alpha\circ p_\alpha$.  It turns out that a uniform choice
$p_\alpha(x)\eqdef 6x+5$ fits those requirements---it is a linear
function in $\FGH{0}$ and $F_\alpha\leq A_\alpha\circ p_\alpha$ as
shown in \autoref{prop-ack}\mbox{---,} thus:
\begin{theorem}\label{th-ack}
  For all $\alpha>0$, $\mathbf{A}_\alpha=\F\alpha$.
\end{theorem}

\subsubsection{Relativised Hierarchies}\label{ssub-rel}
Another means of defining a variant of the fast-growing functions is
to pick a different definition for $F_0$: recall the relativised
fast-growing functions employed in \eqref{eq-rfg-def}. The
corresponding relativised complexity classes are then defined by
\begin{equation}
  \F{h,\alpha}\eqdef\bigcup_{p\in\FGH{<\alpha}}\CC{DTime}\left(F_{h,\alpha}(p(n))\right).
\end{equation}
It is easy to check that,
if $g\leq h$, then $F_{g,\alpha}\leq F_{h,\alpha}$ for all $\alpha$.  Because we
assumed $h$ to be strictly increasing, this entails $F_\alpha\leq
F_{h,\alpha}$, and we have the inclusion $\F\alpha\subseteq
\F{h,\alpha}$ for all strictly increasing $h$.

The converse inclusion does not hold, since for instance $F_{h,1}$ is
non elementary for $h(x)=2^x$.  Observe however that, in this
instance, $h\leq F_2$, and we can see that $F_{F_2,k}=F_{2+k}$ for all
$k$ in $\+N$.  This entails that $\F{h,1}\subseteq\F3$ for $h(x)=2^x$.
Thus, when working with relativised classes, one should somehow
``offset'' the ordinal index by an appropriate amount.

There is nevertheless a difficulty with relativised functions: it is
rather straightforward to show that $F_{h,\alpha}\leq
F_{\beta+\alpha}$ if $h\leq F_\beta$, \emph{assuming} that the direct
sum $\beta+\alpha$ does not ``discard'' any summand from the \ref{CNF}
of $\beta$; e.g.\ $F_{F_1,k}=F_{k+1}$ and
$F_{F_\omega,\omega}=F_{\omega\cdot 2}$.  Observe however that
$F_{F_1,\omega}(x)=F_{F_1,x+1}(x)=F_{x+2}(x)>F_{x+1}(x)=F_\omega(x)$.
Thanks to the closure of $\F\alpha$ under reductions in $\FGH{<\alpha}$,
this issue can be solved by composing with an appropriate function,
e.g.\ $F_{F_1,\omega}(x)\leq F_\omega(x+1)$.  This idea is formalised
in \autoref{app-rel}, and allows to show:
\begin{theorem}\label{th-rel}\renewcommand{\theenumi}{\roman{enumi}}
  Let $h{:}\,\+N\to\+N$ be a strictly increasing function and
  $\alpha,\beta$ be two ordinals.
  \begin{enumerate}\renewcommand{\labelenumi}{(\roman{enumi})}
  \item\label{th-rel-1} If $h\in\FGH{\beta}$, then
    $\F{h,\alpha}\subseteq\F{\beta+1+\alpha}$.
  \item\label{th-rel-2} If $h\leq F_\beta$, then
    $\F{h,\alpha}\subseteq\F{\beta+\alpha}$.
  \end{enumerate}
\end{theorem}
\begin{proof}
  For \eqref{th-rel-1}, if $h$ is in $\FGH\beta$, then there exists
  $x_h$ in $\+N$ such that, for all $x\geq x_h$, $h(x)\leq
  F_{\beta+1}(x)$ \citep[\lemmaautorefname~2.7]{lob70}.  By
  \autoref{lem-rel}, this entails that for all $x\geq x_h$,
  $F_{h,\alpha}(x)\leq F_{\beta+1+\alpha}(F_\gamma(x))$ for some
  $\gamma<\beta+1+\alpha$.  Define the function $f_h$ by $f_h(x)\eqdef
  x+x_h$; then for all $x$, $F_{h,\alpha}(x)\leq
  F_{h,\alpha}(f_h(x))\leq F_{\beta+1+\alpha}(F_\gamma(f_h(x)))$.
  Observe that $F_\gamma\circ f_h$ is in $\FGH{<\beta+1+\alpha}$, thus
  $\F{h,\alpha}\subseteq\F{\beta+1+\alpha}$.

  For \eqref{th-rel-2}, if $\beta+\alpha=0$, then $\beta=\alpha=0$,
  thus $h(x)=x+1$ since it has to be strictly increasing, and
  $F_{h,0}=F_0$.  Otherwise, \autoref{lem-rel} shows that
  $F_{h,\alpha}\leq F_{\beta+\alpha}\circ F_\gamma$ for some
  $\gamma<\beta+\alpha$.  Observe that $F_\gamma$ is in
  $\FGH{<\beta+\alpha}$, thus $\F{h,\alpha}\subseteq\F{\beta+\alpha}$.
\end{proof}

The statement of \autoref{th-rel} is somewhat technical, but
easy to apply to concrete situations; for instance:
\begin{corollary}\label{cor-rel}
  Let $h{:}\,\+N\to\+N$ be a strictly increasing primitive recursive
  function and $\alpha\geq\omega$.  Then $\F{h,\alpha}=\F\alpha$.
\end{corollary}
\begin{proof}
  The function $h$ is in $\FGH{k}$ for some $k<\omega$, thus
  $\F{h,\alpha}\subseteq\F{k+1+\alpha}=\F\alpha$ by \autoref{th-rel}.
  Conversely, since $h$ is strictly increasing,
  $\F\alpha\subseteq\F{h,\alpha}$.
\end{proof}

\subsubsection{Fundamental Sequences}  Our last example of a minor
variation is to change the assignment of fundamental sequences.
Instead of the standard assignment of \eqref{eq-fund-def}, we posit a
monotone function $s{:}\,\+N\to\+N$ and consider the assignment
\begin{align}
  \label{eq-fund-s}
  (\gamma+\omega^{\beta+1})(x)_s&\eqdef\gamma+\omega^\beta\cdot s(x)\;,&
  (\gamma+\omega^{\lambda})(x)_s&\eqdef\gamma+\omega^{\lambda(x)_s}\;.
\end{align}
Thus the standard assignment in \eqref{eq-fund-def} is obtained as the
particular case $s(x)=x+1$.  As previously, this gives rise to new
fast-growing functions
\begin{align}
  F_{0,s}(x)&\eqdef x+1\;,\quad F_{\alpha+1,s}(x)\eqdef
  F^{s(x)}_{\alpha,s}(x)\;,\quad F_{\lambda,s}(x)\eqdef F_{\lambda(x)_s,s}(x)
\shortintertext{and complexity classes}
  \F{\alpha,s}&\eqdef\bigcup_{p\in\FGH{<\alpha}}\CC{DTime}\left(F_{\alpha,s}(p(n))\right).
\end{align}

We obtain similar results with non standard fundamental sequences as
with relativised hierarchies (thus also yielding a statement similar
to that of \autoref{cor-rel}):
\begin{theorem}\label{th-fund}\renewcommand{\theenumi}{\roman{enumi}}
  Let $s{:}\,\+N\to\+N$ be a strictly increasing function and
  $\alpha,\beta$ be two ordinals.
  \begin{enumerate}\renewcommand{\labelenumi}{(\roman{enumi})}
  \item\label{th-fund-1} If $s\in\FGH{\beta}$, then
    $\F{\alpha,s}\subseteq\F{\beta+1+\alpha}$.
  \item\label{th-fund-2} If $s\leq F_\beta$, then
    $\F{\alpha,s}\subseteq\F{\beta+\alpha}$.
  \end{enumerate}
\end{theorem}
\begin{proof}
  By applying \autoref{th-rel} alongside \autoref{lem-fund}.
\end{proof}

The case where $s$ is the identity function $\mathrm{id}(x)\eqdef x$
is fairly common in the literature; we obtain in this particular case:
\begin{corollary}
  For all $\alpha$, $\F{\alpha,\mathrm{id}}=\F{\alpha}$.
\end{corollary}
\begin{proof}
  By \autoref{th-fund} and since $\mathrm{id}\leq F_0$, we have the
  inclusion $\F{\alpha,\mathit{id}}\subseteq\F\alpha$.  The converse
  inclusion stems from $F_\alpha\leq F_{\alpha,\mathrm{id}}\circ F_0$,
  as can be seen by transfinite induction over $\alpha$ (see
  \autoref{lem-id-fund}).
\end{proof}

\subsection{Computational Models and Reductions}\label{sub-redc}
In order to be used together with reductions in $\FGH{<\alpha}$, the
classes $\F\alpha$ need to be closed under such functions.  The main
technical lemma to this end states:

\begin{lemma}\label{th-redc}
  Let $f$ and $f'$ be two functions in $\FGH{<\alpha}$.  Then there
  exists $p$ in $\FGH{<\alpha}$ such that $f\circ
  F_\alpha\circ f'\leq F_\alpha\circ p$.
\end{lemma}
\begin{proof}
  By \autoref{cor-comp}, we know that there exists $g$ in
  $\FGH{<\alpha}$ such that $f\circ F_\alpha\leq F_\alpha\circ g$.  We
  can thus define $p\eqdef g\circ f'$, which is also in
  $\FGH{<\alpha}$ since the latter is closed under composition, to
  obtain the statement.
\end{proof}

\subsubsection{Computational Models}
Note that because we assume $\alpha\geq 3$, $\FGH{<\alpha}$ contains
all the elementary functions, thus \autoref{th-redc} also entails the
robustness of the $\F\alpha$ classes under changes in the model of
computation---e.g.\ RAM vs. Turing machines vs. Minsky machines,
deterministic or nondeterministic or alternating---or the type of
resources under consideration---time or space; e.g.
\begin{equation}\label{eq-fast-space}
  \F\alpha=\bigcup_{p\in\FGH{<\alpha}}\CC{NTime}\big(F_\alpha(p(n))\big)=\bigcup_{p\in\FGH{<\alpha}}\CC{Space}\big(F_\alpha(p(n))\big)\;.
\end{equation}

\subsubsection{Many-One Reductions}
For a function $f{:}\,\+N\to\+N$ and two languages $A$ and $B$, we say
that $A$ \emph{many-one} reduces to $B$ in time $f(n)$, written
$A\leq_m^f B$, if there exists a Turing transducer $T$ working in
deterministic time $f(n)$ such that, for all $x$, $x$ is in $A$ if and
only if $T(x)$ is in $B$.  For a class of functions $\?C$, we write
$A\leq_m^\?C B$ if there exists $f$ in $\?C$ such that $A\leq_m^f B$.
As could be expected given the definitions, each class $\F\alpha$ is
closed under many-one $\FGH{<\alpha}$ reductions:
\begin{theorem}\label{th-red-m}
  Let $A$ and $B$ be two languages.  If
  $A\leq_m^{\FGH{<\alpha}}B$ and $B\in\F\alpha$, then $A\in\F\alpha$.
\end{theorem}
\begin{proof}
  By definition, $A\leq_m^{\FGH{<\alpha}}B$ means that there exists a
  Turing transducer $T$ working in deterministic time $f(n)$ for some
  $f$ in $\FGH{<\alpha}$; note that this implies that the function
  implemented by $T$ is also in $\FGH{<\alpha}$ by \eqref{eq-honest}.
  Furthermore, $B\in\F\alpha$ entails the existence of a Turing
  machine $M$ that accepts $x$ if and only if $x$ is in $B$ and works
  in deterministic time $F_\alpha(p(n))$ for some $p$ in
  $\FGH{<\alpha}$.  We construct $T(M)$ a Turing machine which, given
  an input $x$, first computes $T(x)$ by simulating $T$, and then
  simulates $M$ on $T(x)$ to decide acceptance; $T(M)$ works in
  deterministic time $f(n)+F_\alpha(p(T(n)))$, which shows that $A$ is
  in $\F\alpha$ by \autoref{th-redc}.
\end{proof}

\subsubsection{Turing Reductions}\label{ssub-red}
We write similarly that $A\leq_T^f B$ if there exists a Turing machine
for $A$ working in deterministic time $f(n)$ with oracle calls to $B$,
and $A\leq_T^{\?C}B$ if there exists $f$ in $\?C$ such that
$A\leq_T^{f}B$.  It turns out that Turing reductions in
$\FGH{<\alpha}$ can be used instead of many-one reductions:
\begin{theorem}\label{th-red-T}
  Let $\alpha\geq 3$ and $A$ and $B$ be two languages.  If
  $A\leq_T^{\FGH{<\alpha}} B$ and $B\in\F\alpha$, then $A\in\F\alpha$.
\end{theorem}
\begin{proof}
  It is a folklore result on queries in recursion theory that, if
  $A\leq_T^f B$, then $A\leq_m^{2^f}B^{\rm tt}$ where $2^f(n)\eqdef
  2^{f(n)}$ and $B^{\rm tt}$ is the \emph{truth table} version of the
  language $B$, which evaluates a Boolean combination of queries
  ``$x\in B$.''  Indeed, we can easily simulate the oracle machine for
  $A$ using a nondeterministic Turing transducer also in time $f(n)$
  that guesses the answers of the $B$ oracle and writes a conjunction
  of checks ``$x\in B$'' or ``$x\not\in B$'' on the output, to be
  evaluated by a $B^{\rm tt}$ machine.  This transducer can be
  determinised by exploring both outcomes of the oracle calls, and
  handling them through disjunctions in the output; it now works in
  time $2^f(n)$.

  Since $\alpha\geq 3$ and $f$ is in $\FGH{<\alpha}$,
  $2^f$ is also in $\FGH{<\alpha}$.  Furthermore, since $B$ is in
  $\F\alpha$, $B^{\rm tt}$ is also in $\F\alpha$.  The statement then
  holds by \autoref{th-red-m}.
\end{proof}

%% file: sec-strict.tex
The purpose of this section is to establish the strictness of the
$(\F\alpha)_\alpha$ hierarchy (\autoref{sub-strict}).  As a first
step, we prove that the $F_\alpha$ functions are ``elementarily''
constructible (\autoref{sub-constr}), which is of independent interest
for combinatorial algorithms in the line of that
of \autoref{ssub-algo}.  We end this section with a remark on the case
$\alpha=2$ (\autoref{sub-2}).

\subsection{Elementary Constructivity}\label{sub-constr}
\input{sec-constr} The functions $F_\alpha$ are known to be
\emph{honest}, i.e.\ to be computable in time
$\FGH\alpha$~\citep{wainer70,fairtlough98}.  This is however not tight
enough for their use in length function theorems, as in
\autoref{ssub-algo}, where we want to compute their value in time
elementary in $F_\alpha$ itself.  Formally, we call a function $f$
\emph{elementarily constructible} if there exists an elementary
function $e$ in $\CC{FELem}=\FGH{<3}^\ast$ such that $f(n)$ can be
computed in time $e(f(n))$ for all $n$.

We present the statement in the more general case of relativised
fast-growing functions, defined in \eqref{eq-rfg-def} and discussed in
\autoref{ssub-rel}; since $F_0(x)=x+1$ is elementarily constructible,
this yields the result that all the $F_\alpha$ functions are
elementarily constructible:

\begin{theorem}\label{th-constr}
  Let $h{:}\,\+N\to\+N$ be an elementarily constructible strictly
  increasing function and $\alpha$ be an ordinal, then $F_{h,\alpha}$
  is also elementarily constructible.
\end{theorem}
\begin{proof}
  Assume that $h(n)$ can be computed in time $e(h(n))$ for some
  fixed elementary monotone function $e$.
  \autoref{prop-constr} shows that $F_{h,\alpha}$ can be computed in
  time $O(f(F_{h,\alpha}(n)))$ for the elementary function $f(x)\eqdef
  x\cdot(p\circ G_{\omega^\alpha}(x))+e(x))$, where $p\circ G_{\omega^\alpha}$ is an
  elementary function that takes the cost of manipulating (an encoding
  of) the ordinal indices into account.  \autoref{th-redc} then yields
  the result.
\end{proof}

\subsection{Strictness}\label{sub-strict}
Let us introduce yet another generalisation of the $(\F\alpha)_\alpha$
classes, which will allow for a characterisation of the
$(\FGH{\alpha}^\ast)_\alpha^{~}$ and $(\FGH{<\alpha}^\ast)_\alpha^{~}$
classes.  For an ordinal $\alpha$ and a finite $c>0$, define
\begin{equation}\label{eq-def-Fc}
  \F\alpha^c\eqdef \bigcup_{p\in\FGH{<\alpha}}\CC{DTime}\big(F_\alpha^c(p(n))\big)\;.
\end{equation}
Thus $\F\alpha$ as defined in \eqref{eq-fast-class} corresponds to the
case $c=1$.

\begin{proposition}\label{prop-Fc}
  For all $\alpha\geq
  2$, $$\FGH{\alpha}^\ast=\bigcup_{c}\F\alpha^c\;.$$
\end{proposition}
\begin{proof}
  The left-to-right inclusion is immediate by definition of
  $\FGH{\alpha}^\ast$ in \eqref{eq-fgh-dec}.  The converse inclusion
  stems from the fact that if $p$ is in $\FGH{\beta}$ for some
  $\beta<\alpha$, then there exists $d$ such that $p\leq
  F_\alpha^d$~\citep[\theoremautorefname~2.10]{lob70}, hence
  $F_\alpha^c\circ p\leq F_\alpha^{c+d}$ by monotonicity of
  $F_\alpha$.
\end{proof}

Let us prove the strictness of the $(\F\alpha^c)_{c,\alpha}$
hierarchy.  By \autoref{prop-Fc} it will also prove that of
$(\FGH\alpha^\ast)_\alpha^{~}$ along the way (note that it is not
implied by the strictness of $(\FGH\alpha)_\alpha$, since it would be
conceivable that none of the separating examples would be
$\{0,1\}$-valued):
\begin{theorem}[Strictness]\label{th-strict}
  For all $c>0$ and $2\leq\beta<\alpha$,
  $$\F\beta^c\subsetneq\F\beta^{c+1}\subsetneq\F\alpha\;.$$
\end{theorem}
\begin{proof}[Proof of $\F\beta^{c+1}\subsetneq\F\alpha$]
  Consider first a language $L$ in $\F\beta^{c+1}$, accepted by a
  Turing machine working in time $F_\beta^{c+1}\circ p$ for some $p$
  in $\FGH{<\beta}$ that we can assume to be monotone.  Since
  $\beta<\alpha$ and $F_\beta^{c+1}\circ p$ is in $\FGH\beta$, there
  exists $n_0$ such that, for all $n\geq n_0$,
  $F_\beta^{c+1}(p(n))\leq F_\alpha(n)$, hence for all $n$,
  $F_\beta^{c+1}(p(n))\leq F_\beta^{c+1}(p(n+n_0))\leq
  F_\alpha(n+n_0)$ by monotonicity and expansivity of $F_\beta$.
  Observe that the function $n\mapsto n_0+n$ is in
  $\FGH0\subseteq\FGH{<\alpha}$, thus $L$ also belongs to $\F\alpha$.

  The strictness of the inclusion can be shown by a straightforward
  diagonalisation argument.  Define for this the language
  \begin{equation}\label{eq-strict-proof}
    L_\alpha\eqdef\{\tup M\#x\mid M \text{ accepts $x$ in
      $F_\alpha(|x|)$ steps}\}
  \end{equation}
  where $\tup M$ denotes a description of the Turing machine $M$ and
  $\#$ is a separator.  Then, by \autoref{th-constr}, $L_\alpha$
  belongs to $\F\alpha$, thanks to a Turing machine that first
  computes $F_\alpha$ in time $F_\alpha\circ e$ for some elementary
  function $e$, and then simulates $M$ in time elementary in
  $F_\alpha\circ e$.  Assume now for the sake of contradiction that
  $L_\alpha$ belongs to $\F\beta^{c+1}$, i.e.\ that there exists some
  $c$ and some Turing machine $K$ that accepts $L_\alpha$ in time
  $F_\beta^{c+1}$.  Again, since $\beta<\alpha$ and $F_\beta^{c+1}\circ F_1$
  is in $\FGH\beta$, there exists $n_0$ such that, for all $n\geq
  n_0$, $F_\beta^{c+1}(2n+1)\leq F_\alpha(n)$.  We exhibit a new Turing
  machine $N$
  \begin{enumerate}
  \item that takes as input the description $\tup M$ of a Turing machine
  and simulates $K$ on $\tup M\#\tup M$ but accepts if and
  only if $K$ rejects, and
  \item we ensure that a description $\tup N$ of $N$ has size $n\geq n_0$.
  \end{enumerate}
  Feeding this description $\tup N$ to $N$, it runs in time
  $F_\beta^{c+1}(2n+1)\leq F_\alpha(n)$, and we obtain a contradiction whether
  it accepts or not:
  \begin{itemize}
  \item if $N$ accepts, then $K$ rejects $\tup N\#\tup N$ which is
    therefore not in $L_\alpha$, thus $N$ does not accept $\tup N$ in
    at most $F_\alpha(n)$ steps, which is absurd;
  \item if $N$ rejects, then $K$ accepts $\tup N\#\tup N$ which is
    therefore in $L_\alpha$, thus $N$ accepts $\tup N$ in at most
    $F_\alpha(n)$ steps, which is absurd.\qedhere
  \end{itemize}
\end{proof}
\begin{proof}[Proof of $\F\beta^c\subsetneq\F\beta^{c+1}$]
  Similar to the previous proof; picking $F_\beta^{c+1}$ as the time
  bound instead of $F_\alpha$ in \eqref{eq-strict-proof} suffices to
  establish strictness.
\end{proof}

By \autoref{prop-Fc}, a first consequence of \autoref{th-strict} is
that
\begin{equation}
   \FGH{\beta}^\ast\subsetneq\F\alpha
\end{equation}
for all $2\leq \beta<\alpha$.  Another consequence is that
$(\F\alpha)_\alpha$ ``catches up'' with $(\FGH\alpha^\ast)_\alpha^{~}$ at
every limit ordinal:
\begin{corollary}\label{cor-strict}
  Let $\lambda$ be a limit ordinal, then
\begin{equation*}
  \FGH{<\lambda}^\ast=\bigcup_{\beta<\lambda}\F\beta\subsetneq\F\lambda\;.
\end{equation*}
\end{corollary}
\begin{proof}
  The equality $\FGH{<\lambda}^\ast=\bigcup_{\beta<\lambda}\F\beta$
  and the inclusion $\FGH{<\lambda}^\ast\subseteq\F\lambda$ can be checked
  by considering a problem in some $\FGH{\beta}^\ast$ for
  $\beta<\lambda$: it is in $\F{\beta}^c$ for some $c>0$
  by \autoref{prop-Fc}, hence in $\F{\beta+1}$ with $\beta+1<\lambda$
  by \autoref{th-strict}, and therefore in $\F\lambda$ again
  by \autoref{th-strict}.  Regarding the strictness of the inclusion,
  assume for the sake of contradiction
  $\F\lambda\subseteq\bigcup_{\beta<\lambda}\F\beta$: this would entail
  $\F\lambda\subseteq\F\beta$ for some $\beta<\lambda$,
  violating \autoref{th-strict}.
\end{proof}\noindent
\autoref{cor-strict} yields %
another characterisation of the primitive-recursive and
multiply-recursive problems as
\begin{align}
  \CC{PR}&=\bigcup_{k}\F{k}\;,&
  \CC{MR}&=\bigcup_{k}\F{\omega^k}\;.
\end{align}

Note that strictness implies that there are no
``$\FGH{\alpha}^\ast$-complete'' problems under $\FGH{<\alpha}$
reductions, since by \autoref{prop-Fc} such a problem would
necessarily belong to some $\F{\alpha}^c$ level, which would in turn
entail the collapse of the $(\F{\alpha}^c)_c^{~}$ hierarchy at the
$\F{\alpha}^c$ level and contradict \autoref{th-strict}.

Similarly, fix a limit ordinal $\lambda$ and some reduction class
$\FGH{\alpha}$ for some $\alpha<\lambda$: there cannot be any
meaningful ``$\FGH{<\lambda}^\ast$-complete'' problem under
$\FGH\alpha$ reductions, since such a problem would be in
$\FGH\beta^\ast$ for some $\alpha<\beta<\lambda$, hence contradicting
the strictness of the $(\FGH\beta^\ast)_{\beta<\alpha}^{~}$ hierarchy; in
particular, there are no ``\pr-complete'' nor ``\CC{MR}-complete''
problems.

\subsection{The Case {\boldmath $\alpha=2$}\eatpunct.}\label{sub-2}
\ifsubmission This case \fi is a bit particular.  We did not consider
it in the rest of the paper (nor the other cases for $\alpha<2$)
because it does not share the usual characteristics of the
$(\F\alpha)_\alpha$: for instance, the model of computation and the
kind of resources become important, as
\begin{equation}
  \F2\eqdef \bigcup_{p\in\FGH 1}\CC{DTime}\big(F_2(p(n))\big)
\end{equation}
would a priori be different if we were to define it through \CC{NTime}
or \CC{DSpace} computations; the following results are
artifacts of this particular choice of a definition.

\subsubsection{Recursion Schemes}\label{ssub-rec}In order to define
$\F2$ fully we need the original definition of the extended
Grzegorczyk hierarchy $(\FGH\alpha)_\alpha$ by \citet{lob70}---the
characterisation in \eqref{eq-fgh-char} is only correct for $\alpha\geq
2$.  This definition is based on the closure of a set of initial
functions under the operations of \emph{substitution}
and \emph{limited primitive recursion}.  More precisely, the set of
initial functions at level $\alpha$ comprises the
constant \define{zero function} $0$, the
\define{sum function} $+{:}\,x_1,x_2\mapsto x_1+x_2$, the
\define{projections} $\pi^n_i{:}\,x_1,\dots,x_n\mapsto x_i$ for all
$0<i\leq n$, and the fast-growing function $F_\alpha$.  New functions
are added to form the class $\FGH\alpha$ through two operations:
\begin{description}
  \item[substitution] if $h_0,h_1,\dots,h_p$
    belong to the class, then so does $f$ if
    \begin{equation*}
      f(x_1,\dots,x_n)=h_0(h_1(x_1,\dots,x_n),\dots,h_p(x_1,\dots,x_n))
      \:,
    \end{equation*}
  \item[limited primitive recursion]
    if $h_0$, $h_1$, and $g$ belong to the class, then so
    does $f$ if
    \begin{align*}
      f(0,x_1,\dots,x_n)&=h_0(x_1,\dots,x_n)
      \:,\\
      f(y+1,x_1,\dots,x_n)&=h_1(y,x_1,\dots,x_n,f(y,x_1,\dots,x_n))
      \:,\\
      f(y,x_1,\dots,x_n)&\leq g(\max\{y,x_1,\dots,x_n\})\;.
    \end{align*}
  \end{description}
Observe that primitive recursion is defined by ignoring the last
\emph{limitedness} condition in the previous definition.  See
the survey by \citet{clote99} on the relationships between
machine-defined and recursion-defined complexity classes.

\subsubsection{Linear Exponential Time}\label{ssub-lin}
Let us focus for now on $\FGH1$, which is the class of reductions used
in $\F2$.  First note that the successor function $\suc(x)=
x+1=x+F_1(0)$ belongs to $\FGH1$.

Call a function $f$ \emph{linear} if there exists a constant $c$ such
that $f(x_1,\dots,x_n)\leq c\cdot\max_i x_i$ for all $x_1,\dots,x_n$.
Observe that, for all $c$, the function $f_c(x)\eqdef c\cdot x$ is in
$\FGH 1$ since $f_c(0)=0$, $f_c(x+1)=\suc^c(0)+f_c(x)$, and $f_c(x)\leq
F_1^c(x)$; thus any linear function is bounded above by a function in
$\FGH1$.  Conversely, if $f$ is in $\FGH1$, then it is linear: this is
true of the initial functions, and preserved by the two operations of
substitution and limited primitive recursion.\footnote{Thus
$\FGH1\subsetneq\GH2$: the latter additionally contains the function
$x,y\mapsto (x+1)\cdot(y+1)$ as an initial function, and is equal to
$\CC{FLinSpace}$ \citep[\theoremautorefname~3.36]{ritchie63,clote99}.}

This entails that $\F2$ matches a well-known complexity class, since
furthermore $F_2(n)=2^{n+1+\log (n+1)}-1$ is in $2^{O(n)}$:
$\F2$ is the \emph{weak} (aka \emph{linear}) exponential-time
complexity class:
\begin{equation}
  \F2=\CC{E}\eqdef\CC{DTime}(2^{O(n)})\;.
\end{equation}

%% file: sec-bestiary.tex
\ifsubmission\makeatletter\let\@period\relax\makeatother\fi
\section{A Short Catalogue}
\label{sec-bestiary}

Our introduction of the fast-growing complexity classes is motivated
by \emph{already known} decidability problems, arising for instance in
logic, verification, or database theory, for which no precise
classification could be provided in the existing hierarchies.  By
listing some of these problems, we hope to initiate the exploration of
this mostly uncharted area of complexity, and to foster the use of
reductions from known problems, rather than proofs from Turing
machines.  The following catalogue of complete problems does not attempt
to be exhaustive; \citet{friedman99} for instance presents many
problems ``of enormous complexity.''

Because examples for \textsc{Tower} are well-known and abound in the
literature, starting with a 1975 survey by
\citet{meyer74},\footnote{Of course \citeauthor{meyer74} does not
  explicitly state $\Tow$-completeness, but it follows immediately from
  the lower and upper bounds he provides.} we rather focus on the non
primitive-recursive levels, i.e.\ the $\F\alpha$ for
$\alpha\geq\omega$.  Interestingly, all these examples rely for their
upper bound on the existence of some well-quasi-ordering (of
\emph{maximal order type} $\omega^\alpha$~\citep[see][]{dejongh77}), and
on a matching length function theorem.

\subsection{{\boldmath $\F\omega$}-Complete Problems}
\label{sec-Fo}

We gather here some of the decision problems known to be \Ack-complete
at the time of this writing.  The common trait of all these problems
is their reliance on Dickson's Lemma over $\+N^d$ for some $d$ for
decidability, and on the associated length function
theorems \citep{mcaloon,clote,FFSS-lics2011,abriola} for \Ack\ upper
bounds. %

\subsubsection{Vector Addition Systems\nopunct}\label{ssub-vas}
(VAS, and equivalently Petri nets), provided the first known
Ackermannian decision problem: \probref{pb-fcp}.

A $d$-dimensional VAS is a pair $\tup{\vec{v}_0,\vec{A}}$ where
$\vec{v}_0$ is an initial configuration in $\+N^d$ and $\vec{A}$ is a
finite set of transitions in $\+Z^d$.  A transition $\vec{u}$ in
$\vec{A}$ can be applied to a configuration $\vec{v}$ in $\+N^d$ if
$\vec{v}'=\vec{v}+\vec{u}$ is in $\+N^d$; the resulting configuration
is then $\vec{v}'$.  The complexity of decision problems for VAS
usually varies from \textsc{ExpSpace}-complete
\citep{lipton76,rackoff78,mfcs/BlockeletS11} to $\F\omega$-complete
\citep{fct,jancar} to undecidable \citep{hack76,jancar95b}, via a key
problem, whose exact complexity is unknown: \pbfont{VAS
  Reachability}
\citep{mayr,kosa,lambert,leroux-popl2011,leroux15}.

\begin{problem}[FCP]{Finite Containment
    Problem}\label{pb-fcp}\hfill
\begin{description}[topsep=0pt,itemsep=0pt,partopsep=0pt,parsep=0pt]
\item[instance:] Two VAS $\?V_1$ and $\?V_2$ known to have finite sets
  $\mathrm{Reach}(\?V_1)$ and $\mathrm{Reach}(\?V_2)$ of reachable
  configurations.
\item[question:] Is $\mathrm{Reach}(\?V_1)$ included in
  $\mathrm{Reach}(\?V_2)$?
\item[lower bound:] \citet{fct}, from an $F_\omega$-bounded version of
  \pbfont{Hilbert's Tenth Problem}.  A simpler reduction is given by
  \citet{jancar} from $F_\omega\text-\mathsf{MM}$ the halting problem of
  $F_\omega$-bounded Minsky machines.
\item[upper bound:] Originally \citet{mcaloon} and \citet{clote}, or
  more generally using length function theorems for Dickson's
  Lemma \citep{FFSS-lics2011,abriola}.
\item[comment:] Testing whether the set of reachable configurations of
  a VAS is finite is \CC{ExpSpace}-complete
  \citep{lipton76,rackoff78}.  \probref{pb-fcp} has been generalised by
  \citet{jancar} to a large range of behavioural relations between
  two VASs.  Without the finiteness condition, these questions are
  undecidable \citep{hack76,jancar95b,jancar}.
\end{description}
\end{problem}

An arguably simpler problem on vector addition systems has recently
been shown to be \Ack-complete by \citet{hofman14}.  A \emph{labelled
  vector addition system with states} (VASS)
$\?V=\tup{Q,\Sigma,d,T,q_0,\vec{v}_0}$ is a VAS extended with a finite
set $Q$ of control states that includes a distinguished initial state
$q_0$.  The transitions in $T$ of such systems are furthermore
labelled with symbols from a finite alphabet $\Sigma$: transitions are
then defined as quadruples $q\xrightarrow{a,\vec u}q'$ for $a$ in
$\Sigma$ and $\vec u$ in $\+Z^d$.  Such a system defines an infinite
labelled transition system $\tup{Q\times\+N^d,\to,(q_0,\vec{v}_0)}$
where $(q,\vec v)\xrightarrow{a}(q',\vec v+\vec u)$ if
$q\xrightarrow{a,\vec u}q'$ is in $T$ and $\vec v+\vec u\geq\vec 0$.
The \emph{set of traces} of $\?V$ is the set of finite sequences
$L(\?V)\eqdef\{a_1\cdots a_n\in\Sigma^\ast\mid\exists (q,\vec v)\in
Q\times\+N^d.(q_0,\vec{v}_0)\xrightarrow{a_1\cdots a_n}(q,\vec v)\}$.
\begin{problem}[1VASSU]{One-Dimensional VASS Universality}\label{pb-1vassu}\hfill
\begin{description}[topsep=0pt,itemsep=0pt,partopsep=0pt,parsep=0pt]
\item[instance:] A one-dimensional labelled VASS
  $\?V=\tup{Q,\Sigma,1,T,q_0,\vec{x}_0}$.
\item[question:] Does $L(\?V)=\Sigma^\ast$, i.e.\ is every finite
  sequence over $\Sigma$ a trace of $\?V$?
\item[lower bound:] \citet{hofman14} by reduction from reachability in
  gainy counter machines, see \probref{pb-lcn}.
\item[upper bound:] \citet{hofman14} using length function theorems for
  Dickson's Lemma.%
\item[comment:] One-dimensional VASS are also called ``one counter
  nets'' in the literature.  More generally, the \emph{inclusion}
  problem $L\subseteq L(\?V)$ for some rational language $L$ is still
  \Ack-complete.
\end{description}
\end{problem}

\subsubsection{Unreliable Counter Machines.}
A \emph{lossy counter machine} (LCM) is syntactically a Minsky
machine, but its operational semantics are different: its counter
values can decrease nondeterministically at any moment during
execution.  See \autoref{sub-lcm} for details.
\begin{problem}[LCM]{Lossy Counter Machines
    Reachability}\label{pb-lcn}\hfill
\begin{description}[topsep=0pt,itemsep=0pt,partopsep=0pt,parsep=0pt]
\item[instance:] A lossy counter machine $M$ and a configuration
  $\sigma$.
\item[question:] Is $\sigma$ reachable in $M$ with lossy semantics?
\item[lower bound:] \citet{phs-mfcs2010}, by a direct reduction from
  $F_\omega$-bounded Minsky machines.  The first proofs were given
  independently by \citeauthor{urquhart99}
  in~\citeyear{urquhart99}~\citep{urquhart99} and
  \citeauthor{phs-IPL2002} in~\citeyear{phs-IPL2002}~\citep{phs-IPL2002}.
\item[upper bound:] Length function theorem for Dickson's
  Lemma.%
\item[comment:] Completeness also holds for terminating LCMs (meaning
  that every computation starting from the initial configuration
  terminates), coverability in Reset or Transfer Petri nets, and
  for reachability in \emph{gainy} counter machines, where counter
  values can increase nondeterministically.
\end{description}
\end{problem}

\subsubsection{Relevance Logics\nopunct} provide different semantics of
implication, where a fact $B$ is said to follow from $A$, written
``$A\rightarrow B$'', only if $A$ is actually \emph{relevant} in the
deduction of $B$.  This excludes for instance $A\rightarrow(B\rightarrow A)$,
$(A\wedge\neg A)\rightarrow B$, etc.---see \citet{dunn02} for more
details.  Although the full logic $\mathbf{R}$ is undecidable
\citep{urquhart84}, its conjunctive-implicative fragment
$\mathbf{R}_{\rightarrow,\wedge}$ is decidable, and \Ack-complete:
\begin{problem}[CRI]{Conjunctive Relevant
    Implication}\label{pb-rec}\hfill
\begin{description}[topsep=0pt,itemsep=0pt,partopsep=0pt,parsep=0pt]
\item[instance:] A formula $A$ of $\mathbf{R}_{\rightarrow,\wedge}$.
\item[question:] Is $A$ a theorem of $\mathbf{R}_{\rightarrow,\wedge}$?
\item[lower bound:] \citet{urquhart99}, from a variant of
  \probref{pb-lcn}: the emptiness problem of \defstyle{alternating
    expansive counter systems}, for which he proved
    $\F\omega$-hardness directly from $F_\omega\text-\mathsf{MM}$ the
    halting problem in $F_\omega$-bounded Minsky machines.
\item[upper bound:] \citet{urquhart99} using length function theorem
    for Dickson's Lemma.%
\item[comment:] Hardness also holds for any
  intermediate logic between $\mathbf{R}_{\rightarrow,\wedge}$ and
  $\mathbf{T}_{\rightarrow,\wedge}$, which might include some
  undecidable fragments.  The related \emph{contractive propositional
  linear logic} LLC and its additive-multiplicative fragment MALLC are
  also \Ack-complete~\citep{lazic14}.
\end{description}
\end{problem}

\subsubsection{Data Logics \& Register
  Automata\nopunct}\label{sub-datal} are concerned
with structures like words or trees with an additional equivalence
relation over the positions.  The motivation for this stems in
particular from XML processing, where the equivalence stands for
elements sharing the same \emph{datum} from some infinite data domain
$\+D$.  Enormous complexities often arise in this context, both
for automata models (register automata and their variants, when
extended with alternation or histories) and for logics (which include
logics with \emph{freeze} operators and XPath fragments)---the two
views being tightly interconnected.

\begin{problem}[A1RA]{Emptiness of Alternating
    1-Register Automata}\label{pb-ara}\hfill
\begin{description}[topsep=0pt,itemsep=0pt,partopsep=0pt,parsep=0pt]
\item[instance:] An A1RA $\?A$.
\item[question:] Is the data language $L(\?A)$ empty?
\item[lower bound:] \citet{demri09}, from reachability in gainy
  counter machines \probref{pb-lcn}.
\item[upper bound:] \citet{demri09}, by reducing to reachability in gainy
  counter machines \probref{pb-lcn}.
\item[comment:] There exist many variants of the A1RA model, and
  hardness also holds for the corresponding data logics
  \citep[e.g.][]{jurdzinski2007,demri09,FigSeg-mfcs09,tan2010,figueira12,tzelevekos13}.
  See \probref{pb-ata} for the case of linearly ordered data,
  and \probref{pb-odl} for data logics using multiple attributes with a
  hierarchical policy.
\end{description}
\end{problem}

\subsubsection{Metric Temporal Logic\nopunct}
(MTL) allows %
to reason on \defstyle{timed
words} over $\Sigma\times\+R$, where $\Sigma$ is a finite alphabet and
the real values are non decreasing \emph{timestamps} on
events \citep{koymans90}.  When considering infinite timed words, one
usually focuses on \emph{non-Zeno} words, where the timestamps are
increasing and unbounded.  MTL is an extension of linear temporal
logic where temporal modalities are decorated with real intervals
constraining satisfaction; for instance, a timed word $w$ satisfies
the formula $\mathsf{F}_{[3,\infty)}\varphi$ at position $i$, written
$w,i\models\mathsf{F}_{[3,\infty)}\varphi$, only if $\varphi$ holds at
some position $j>i$ of $w$ with timestamp $\tau_j-\tau_i\geq 3$.
The \emph{safety} fragment of MTL restricts the intervals decorating
``until'' modalities to be right-bounded.

\begin{problem}[SMTL]{Satisfiability of Safety
    Metric Temporal Logic}\label{pb-smtl}
\begin{description}[topsep=0pt,itemsep=0pt,partopsep=0pt,parsep=0pt]
\item[instance:] A safety MTL formula $\varphi$.
\item[question:] Does there exist an infinite non-Zeno timed word $w$ s.t.\
  $w,0\models\varphi$?
\item[lower bound:] \citet{lazic13}, by a direct reduction from
  $F_\omega$-bounded Turing machines.
\item[upper bound:] \citet{lazic13} by resorting to length function
  theorems for Dickson's Lemma.
\item[comment:] The complexity bounds are established through
reductions to and from the \emph{fair termination} problem for
insertion channel systems, which \citet{lazic13} show to
be \Ack-complete; see \probref{pb-lcst}.
\end{description}
\end{problem}

\subsubsection{Ground Term Rewriting\nopunct.}
A \emph{ground term rewrite system with state} (sGTRS) maintains a
finite ordered labelled tree along with a control state from some
finite set.  While most questions about ground term rewrite systems
are decidable~\citep{dauchet90}, the addition of a finite set of
control states yields a Turing-powerful formalism.  Formally, a sGTRS
$\tup{Q,\Sigma,R}$ over a ranked alphabet $\Sigma$ and a finite set of
states $Q$ is defined by a finite set of rules $R\subseteq (Q\times
T(\Sigma))^2$ of the form $(q,t)\to(q',t')$ acting over pairs of
states and trees, which rewrite a configuration $(q,C[t])$ into
$(q',C[t'])$ in any context~$C$.

\Citet{hague14} adds \emph{age} labels in $\+N$ to every node of the
current tree.  In the initial configuration, every tree node has age
zero, and at each rewrite step $(q,C[t])\to(q',C[t'])$, in the
resulting configuration the nodes in $t'$ have age zero, and the nodes
in $C$ see their age increment by one if $q\neq q'$ or remain with the
same age as in $(q,C[t])$ if $q=q'$.  A \emph{senescent} sGTRS with
\emph{lifespan} $k$ in $\+N$ restricts rewrites to only occur in
subtrees of age at most $k$, i.e.\ when matching $C[t]$ the age of the
root of $t$ is $\leq k$.
\begin{problem}[SGTRS]{State Reachability in Senescent Ground Term
    Rewrite Systems}\label{pb-sgtrs}\hfill\vspace*{-1em}
\begin{description}[topsep=0pt,itemsep=0pt,partopsep=0pt,parsep=0pt]
\item[instance:] A senescent sGTRS $\tup{Q,\Sigma,R}$ with lifespan
  $k$, two states $q_0$ and $q_f$ in $Q$, and an initial tree $t_0$ in
  $T(\Sigma)$.
\item[question:] Does there exist a tree $t$ in $T(\Sigma)$ such that
  $(q_f,t)$  is reachable from $(q_0,t_0)$?
\item[lower bound:] \citet{hague14}, from coverability in reset
  Petri nets, see~\probref{pb-lcn}.
\item[upper bound:] \citet{hague14}, by reducing to coverability in
reset Petri nets, see~\probref{pb-lcn}.
\end{description}
\end{problem}

\subsubsection{Interval Temporal Logics\nopunct} provide a formal
framework for reasoning about temporal intervals.
\Citet{halpern91interval} define a logic with modalities expressing
the basic relationships that can hold between two temporal intervals,
$\tup{B}$ for ``begun by'', $\tup{E}$ for ``ended by'', and their
inverses $\tup{\bar{B}}$ and $\tup{\bar{E}}$.  This logic, and even
small fragments of it, has an undecidable satisfiability problem, thus
prompting the search for decidable restrictions and variants.
\Citet{montanari10} show that the logic with relations
$A\bar{A}B\bar{B}$---where $\tup{A}$ expresses that the two intervals
``meet'', i.e.\ share an endpoint---, has an $\F\omega$-complete
satisfiability problem over finite linear orders:
\begin{problem}[ITL]{Finite Linear Satisfiability of
    $A\bar{A}B\bar{B}$}\label{pb-itl}\hfill
\begin{description}[topsep=0pt,itemsep=0pt,partopsep=0pt,parsep=0pt]
\item[instance:] An $A\bar{A}B\bar{B}$ formula $\varphi$.
\item[question:] Does there exist an interval structure $\?S$ over
  some finite linear order and an interval $I$ of $\?S$ s.t.\
  $\?S,I\models\varphi$?
\item[lower bound:] \citet{montanari10}, from reachability in lossy
  counter systems~(\probref{pb-lcn}).
\item[upper bound:] \citet{montanari10}, by reducing to reachability
in lossy counter systems~(\probref{pb-lcn}).
\item[comment:] Hardness already holds for the fragments $\bar{A}B$
  and $\bar{A}\bar{B}$ \citep{bresolin2012}.
\end{description}
\end{problem}

\subsection{{\boldmath $\F{\omega^\omega}$}-Complete Problems}
\label{sec-Foo}

The following problems are known to be complete for \hack.  In most
cases they have been proven decidable thanks to Higman's Lemma over
some finite alphabet, and the complexity upper bounds stem from the
length function theorems of \citet{weiermann94,cichon98,SS-icalp2011}.

\subsubsection{Lossy Channel Systems\nopunct}\label{ssub-lcs}
(LCS) are finite labelled transition systems $\tup{Q,M,\delta,q_0}$
where transitions in $\delta\subseteq Q\times \{?,!\}\times M\times Q$
read and write on an unbounded channel.  This would lead to a
Turing-complete model of computation, but the operational semantics of
LCS are ``lossy'': the channel loses symbols in an uncontrolled
manner.  Formally, the configurations of an LCS are pairs $(q,x)$,
where $q$ in $Q$ holds the current state and $x$ in $M^\ast$ holds the
current contents of the channel.  A read $(q,{?}m,q')$ in $\delta$
updates this configuration into $(q,x')$ if there exists some $x''$
s.t.\ $x'\leq_\ast x''$ and $mx''\leq_\ast x$---where $\leq_\ast$
denotes subword embedding---, while a write transition $(q,{!}m,q')$
updates it into $(q',x')$ with $x'\leq_\ast xm$; the initial
configuration is $(q_0,\varepsilon)$, with empty initial channel
contents.

Due to the unboundedness of the channel, there might be infinitely
many configurations reachable through transitions.  Nonetheless, many
problems are decidable \citep{abdulla96b,cece95} using Higman's Lemma
and what would later become known as the theory of
\emph{well-structured transition systems}
(WSTS)~\citep{finkel87c,abdulla2000c,finkel98b}.  LCS are also the
primary source of problems hard for $\F{\omega^\omega}$:

\begin{problem}[LCS]{LCS Reachability}\label{pb-lcs}\hfill
\begin{description}[topsep=0pt,itemsep=0pt,partopsep=0pt,parsep=0pt]
\item[instance:] A LCS and a configuration $(q,x)$ in $Q\times
  M^\ast$.
\item[question:] Is $(q,x)$ reachable from the initial configuration?
\item[lower bound:] \citet{lcs}, by a direct reduction from
  $F_{\omega^\omega}\text-\mathsf{MM}$ the halting problem in
  $F_{\omega^\omega}$-bounded Minsky machines.
\item[upper bound:] \citet{lcs} using the length function theorem
  of \citet{cichon98}, or more generally using length function
  theorems for Higman's Lemma \citep{weiermann94,SS-icalp2011}.
\item[comment:] Hardness holds already for the (semantically defined)
  class of terminating systems, and for reachability
  in \emph{insertion channel systems}, where symbols are
  nondeterministically inserted in the channel at arbitrary positions
  instead of being lost.  The bounds are refined and parametrised in
  function of the size of the alphabet $M$
  in~\citep{fossacs/KarandikarS13}.
\end{description}
\end{problem}

There are many interesting applications of this question; let us
mention one in particular: \citet{tsoreach} show how concurrent
finite programs communicating through \emph{weak} shared
memory---i.e.\ prone to reorderings of read or writes, modelling the
actual behaviour of microprocessors, their instruction pipelines, and
cache levels---have an $\F{\omega^\omega}$-complete control-state
reachability problem, through reductions to and from
\probref{pb-lcs}.

\begin{problem}[LCST]{LCS Termination}\label{pb-lcst}\hfill
\begin{description}[topsep=0pt,itemsep=0pt,partopsep=0pt,parsep=0pt]
\item[instance:] A LCS.
\item[question:] Is every sequence of transitions from the initial
  configuration finite?
\item[lower bound:] \citet{lcs}, by a reduction from terminating
  instances of \probref{pb-lcs}.
\item[upper bound:] Length function theorems for Higman's
Lemma.%
\item[comment:] Unlike \pbfont{Reachability}, \pbfont{Termination} is
  sensitive to switching from lossy semantics to insertion semantics: it
  becomes \textsc{NL}-complete in general~\citep{cece95},
  \textsc{Tower}-complete when the channel system is equipped with
  \emph{channel tests}~\citep{bouyer12}, and \Ack-complete when one asks for
  \emph{fair} non termination, where the channel contents are read
  infinitely often~\citep{lazic13}.
\end{description}
\end{problem}

\subsubsection{Embedding Problems\nopunct}
have been introduced by \citet{pepreg},
motivated by decidability problems in various classes of channel
systems mixing lossy and reliable channels.  These problems are
centred on the subword embedding relation $\leq_\ast$ and called
\pbfont{Post Embedding Problems}.  There is a wealth of variants and
applications, see e.g.\ \citep{CS-fossacs08,KS-csr12,fossacs/KarandikarS13}. 

We give here a slightly different viewpoint, taken from
\citep{barcelo12,fossacs/KarandikarS13}, that uses regular relations
(i.e.\ definable by synchronous finite transducers) and rational
relations (i.e.\ definable by finite transducers):

\begin{problem}[RatEP]{Rational Embedding
    Problem}\label{pb-ratep}\hfill
\begin{description}[topsep=0pt,itemsep=0pt,partopsep=0pt,parsep=0pt]
\item[instance:] A rational relation $R$ included in
  $\Sigma^\ast\times\Sigma^\ast$.
\item[question:] Is $R\cap{\leq_\ast}$ non empty?
\item[lower bound:] \citet{pepreg}, from reachability in lossy
channel systems~(\probref{pb-lcs}).
\item[upper bound:] Length function theorems for Higman's
Lemma.%
\item[comment:] \citet{pepreg} call this problem the \pbfont{Regular
    Post Embedding Problem}, not to be mistaken with \probref{pb-gip}.
  An equivalent presentation uses a rational language $L$ included in
  $\Sigma^\ast$ and two homomorphisms
  $u,v{:}\,\Sigma^\ast\to\Sigma^\ast$, and asks whether there exists
  $w$ in $L$ s.t.\ $u(w)\leq_\ast v(w)$.  The bounds are refined and
  parametrised in function of the size of the alphabet $\Sigma$
  in~\citep{fossacs/KarandikarS13}.
\end{description}
\end{problem}

\begin{problem}[GEP]{Generalised Embedding
    Problem}\label{pb-gip}\hfill
\begin{description}[topsep=0pt,itemsep=0pt,partopsep=0pt,parsep=0pt]
\item[instance:] A regular relation $R$ included in
  $(\Sigma^\ast)^m$ and a subset $I$ of $\{1,...,m\}^2$.
\item[question:] Does there exist $(w_1,\dots,w_m)$ in $R$ s.t.\ for
  all $(i,j)$ in $I$, $w_i\leq_\ast w_j$?
\item[lower bound:] \citet{barcelo12}, from \probref{pb-ratep}.
\item[upper bound:] Length function theorems for Higman's
Lemma.%
\item[comment:] The \pbfont{Regular Embedding Problem}
  (\pbfont{RegEP}) corresponds to the case where $m=2$ and
  $I=\{(1,2)\}$, and is already $\F{\omega^\omega}$-hard; see
  \citep{fossacs/KarandikarS13} for refined bounds.  \Citet{barcelo12}
  use \probref{pb-gip} to show the $\F{\omega^\omega}$-hardness of
  querying graph databases using particular extended conjunctive
  regular path queries.
\end{description}
\end{problem}

\subsubsection{Timed Automata\nopunct} \citep{alur94} are finite
  automata able to recognise timed words.  They are extended
with \emph{clocks} that evolve synchronously through time, and can be
reset and compared against some time interval by the transitions of
the automaton.  The model can be extended with alternation, and is
then called an ATA.  Satisfiability problems for MTL reduce to
emptiness problems for ATAs.  \citet{mtl} and \citet{ata}
prove using WSTS techniques that, in the case of a single clock,
emptiness of ATAs is decidable.  Note that the \emph{safety} fragment
of MTL has an \Ack-complete satisfiability problem,
see \probref{pb-smtl}.
\ifsubmission\relax\else\pagebreak\fi
\begin{problem}[A1TA]{Emptiness of Alternating
    1-Clock Timed Automata}\label{pb-ata}
\begin{description}[topsep=0pt,itemsep=0pt,partopsep=0pt,parsep=0pt]
\item[instance:] An A1TA $\?A$.
\item[question:] Is the timed language $L(\?A)$ empty?
\item[lower bound:] \citet{ata}, from reachability in insertion channel systems
  (\probref{pb-lcs}).
\item[upper bound:] Length function theorems for Higman's
 Lemma.%
\item[comment:] Hardness already holds for universality of
  nondeterministic 1-clock timed automata.
\end{description}
\end{problem}

\begin{problem}[fMTL]{Finite Satisfiability of
    Metric Temporal Logic}\label{pb-mtl}
\begin{description}[topsep=0pt,itemsep=0pt,partopsep=0pt,parsep=0pt]
\item[instance:] An MTL formula $\varphi$.
\item[question:] Does there exist a finite timed word $w$ s.t.\
  $w,0\models\varphi$?
\item[lower bound:] \citet{mtl}, from reachability in insertion channel
systems~(\probref{pb-lcs}).
\item[upper bound:] Length function theorems for Higman's
 Lemma.%
\item[comment:] Satisfiability for infinite timed words is
  undecidable~\citep{ouaknine06}.
\end{description}
\end{problem}
Note that recent work on data automata over linearly ordered domains
has uncovered some strong ties with timed automata
\citep{FHL10,figueira12}.%

\subsubsection{Unordered Data Nets\nopunct} are a generalisation of
Petri nets where each token carries some datum from some infinite data
domain, which can be tested for equality against the data of other
tokens when firing the transitions of the system.  This is a
restriction over the more general \emph{data nets}~\citep{datanets},
where the data domain is deemed to be densely linearly ordered;
see~\probref{pb-enc}.  Like general data nets, unordered data nets
allow so-called ``whole-place'' operations, endowing them with
generalised reset capabilities; the exact complexity of coverability
for \emph{unordered Petri data nets}, where such operations are not
available, is unknown at the moment (\Tow-hardness is shown
by \citet{datanets}).

\begin{problem}[UDN]{Unordered Data Nets
Coverability}\label{pb-udn}\hfill%
\begin{description}[topsep=0pt,itemsep=0pt,partopsep=0pt,parsep=0pt]
\item[instance:] An unordered data net $\?N$ and a place $p$ of the net.
\item[question:] Is there a reachable marking with a least one token
  in $p$?
\item[lower bound:] \citet{rosavelardo14}, by a direct reduction from the
  halting problem in $F_{\omega^{\omega}}$-bounded Minsky
  machines.
\item[upper bound:] \citet{rosavelardo14}, by proving a length
  function theorem for $\mathbb{M}_{\textrm{fin}}(\+N^d)$ the set of
  finite multisets of vectors of naturals, ordered by multiset
  embedding.
\end{description}
\end{problem}
This is the only instance in this list of a \hack-complete problem
that does not explicitly rely on Higman's Lemma.

\subsection{{\boldmath $\F{\omega^{\omega^\omega}}$}-Complete Problems}
\label{sec-Fooo}

Currently, all the known $\F{\omega^{\omega^\omega}}$-complete
problems are related to extensions of Petri nets called
\defstyle{enriched nets}, which include timed-arc Petri nets
\citep{abdulla01}, ordered data nets and ordered Petri data nets
\citep{datanets}, and constrained multiset rewriting systems
\citep{abdulla06}.  Reductions between the different classes of
enriched nets can be found in \citep{abdulla2011,rr-lsv-10-23}.
Defining these families of nets here would take too much space; see
the referenced papers for details.  These models share one
characteristic: they define well-structured transition systems over
finite sequences of vectors of natural numbers, which have an
$\omega^{\omega^{\omega^\omega}}$ maximal order type.

\begin{problem}[ENC]{Enriched Net Coverability}\label{pb-enc}\hfill%
\begin{description}[topsep=0pt,itemsep=0pt,partopsep=0pt,parsep=0pt]
\item[instance:] An enriched net $\?N$ and a place $p$ of the net.
\item[question:] Is there a reachable marking with a least one token
  in $p$?
\item[lower bound:] \citet{HSS-lics2012}, by a direct reduction from the
  halting problem in $F_{\omega^{\omega^\omega}}$-bounded Minsky
  machines.
\item[upper bound:] \citet{HSS-lics2012}, using length function theorems
  for finite sequences of vectors of natural numbers and Higman's
  Lemma~\citep{SS-icalp2011}.
\end{description}
\end{problem}

\subsection{{\boldmath $\F\ezero$}-Complete Problems\nopunct}
Problems complete for $\F\ezero$ are untractable in a distinctive
sense: although there exists a Turing machine able to answer on every
instance, the termination proof of this Turing machine implies a
totality proof for a function akin to $F_{\ezero}$: the latter is
however known to be independent of Peano
Arithmetic~\citep[e.g.][]{fairtlough98}.

\subsubsection{Priority Channel Systems\nopunct} 
are defined similarly to lossy channel
systems (c.f.\ \autoref{ssub-lcs}), but the message alphabet $M$ is
linearly ordered to represent message \emph{priorities}.  Rather than
message losses, the unreliable behaviours are now \emph{message
  supersedings}, i.e.\ applications of the rewrite rules $ab\to b$ for
$b\geq a$ in $M$ on the channel contents.

\begin{problem}[PCS]{PCS Reachability}\label{pb-pcs}\hfill%
\begin{description}[topsep=0pt,itemsep=0pt,partopsep=0pt,parsep=0pt]
\item[instance:] A PCS and a configuration $(q,x)$ in $Q\times
  M^\ast$.
\item[question:] Is $(q,x)$ reachable from the initial configuration?
\item[lower bound:] \citet{HaaseSS13}, by a direct reduction from the
  halting problem in $F_{\ezero}$-bounded Turing machines.
\item[upper bound:] \citet{HaaseSS13}, using length function theorems
for nested applications of Higman's Lemma~\citep{SS-icalp2011}.
\end{description}
\end{problem}

\subsubsection{Nested Counter Systems \& Hierarchical Multi-Attributed
Data Logics.} Finite data words may in general carry several data
values from some infinite data domain in addition to a label from some
finite alphabet.  The satisfiability of data logics over such data
words becomes undecidable, even for the restricted logics discussed
in \autoref{sub-datal}.  However, decidability can be recovered when
the logic is restricted by a hierarchical discipline on its attributes
$\{0,\dots,k\}$, where attribute $i$ can only be tested for equality
on two positions of the word if all the attributes $0,\dots,i-1$ are
also simultaneously tested.
\begin{problem}[LTL$^\downarrow_{[k]}$]{Satisfiability of Freeze LTL with
  Ordered Attributes}\label{pb-odl}\hfill
\begin{description}[topsep=0pt,itemsep=0pt,partopsep=0pt,parsep=0pt]
\item[instance:] A formula $\varphi$ of Freeze LTL with one register
and $k$ hierarchical attributes.
\item[question:] Does there exist a $k$-attributed finite data word
$w$ s.t.\ $w\models\varphi$?
\item[lower bound:] \citet{decker15} by a direct reduction from
$F_{\ezero}$-bounded Minsky machine.
\item[upper bound:] \citet{decker15} by a reduction to reachability
in priority channel systems (\probref{pb-pcs}).
\item[comment:] The complexity bounds are established through the
coverability problem for a class of nested counter
systems~\citep{decker15}.
\end{description}
\end{problem}

%% file: sec-concl.tex
The classical complexity hierarchies are limited to elementary
problems, in spite of a growing number of natural problems that
require much larger computational resources.  We propose in this paper
a definition for fast-growing complexity classes
$(\F\alpha)_{\alpha}$, which provide accurate enough notations for
many non elementary decision problems: they allow to express some
important landmarks, like $\Tow=\F3$, $\Ack=\F\omega$, or
$\hack=\F{\omega^\omega}$, and are close enough to the extended
Grzegorczyck hierarchy so that complexity statements in terms of
$\FGH\alpha$ can often be refined as statements in terms of
$\F\alpha$.  These definitions allow to employ the familiar vocabulary
of complexity theory, reductions and completeness, instead of the more
ad-hoc notions used this far.  This will hopefully foster the reuse of
``canonical problems'' in establishing high complexity results, rather
than proofs from first principles, i.e.\ resource-bounded Turing
machines.

A pattern emerges in the list of known $\F\alpha$-complete problems,
allowing to answer a natural concern already expressed by
\citet{clote}: ``what do complexity classes for such rapidly growing
functions really mean?''  Indeed, beyond the intellectual satisfaction
one might find in establishing a problem as complete for some class,
being $\F\alpha$-complete brings additional information on the problem
itself: that it relies in some essential way on the ordinal
$\omega^\alpha$ being well-ordered.  All the problems in
\autoref{sec-bestiary} match this pattern, as their decision
algorithms rely on well-quasi-orders with maximal order type
$\omega^\alpha$ for their termination, for which length function
theorems then allow to derive $\F\alpha$ bounds.

Finally, we remark that there are currently no known natural problem
of ``intermediate'' complexity, for instance between \elem\ and \Ack,
or between the latter and \hack.  Parametric versions of
\probref{pb-lcn} or \probref{pb-lcs} seem like good candidates for
this, but so far the best lower and upper bounds do not quite
match~\citep[see e.g.][]{fossacs/KarandikarS13}.  It would be
interesting to find examples that exercise the intermediate levels of
the $(\F\alpha)_\alpha$ hierarchy.

%% file: appendix.tex
This section presents the technical background and proofs missing from
the main text.

\subsection{Hardy Functions} Let $h{:}\,\+N\to\+N$ be a strictly
increasing function.  The \emph{Hardy functions}
$(h^\alpha)_{\alpha<\ezero}$ controlled by $h$ are defined
inductively by
\begin{align}\label{eq-hardy-rel}
  h^0(x)&\eqdef x\;,&
  h^{\alpha+1}(x)&\eqdef h^\alpha\left(h(x)\right),&
  h^\lambda(x)&\eqdef h^{\lambda(x)}(x)\;.
\end{align}

A definition related to fundamental sequences is that of the
\emph{predecessor} at $x$ of an ordinal greater than $0$, which
recursively considers the $x$th element in the fundamental sequence of
limit ordinals, until a successor ordinal is found:
\begin{align}\label{eq-pred-def}
  P_x(\alpha+1)&\eqdef \alpha\;,&
  P_x(\lambda)&\eqdef P_x(\lambda(x))\;.
\end{align}
Using predecessors, the definition of the Hardy functions becomes even
simpler: for $\alpha>0$,
\begin{align}
  \label{eq-hardy-pred}
  h^{\alpha}(x)&\eqdef h^{P_x(\alpha)}\left(h(x)\right)\;.
\end{align}
Observe for instance that $h^k(x)$ for some finite $k$ is the $k$th
iterate of $h$.  This intuition carries over: $h^\alpha$ is a
transfinite iteration of the function $h$, using diagonalisation to
handle limit ordinals.  The usual Hardy functions $H^\alpha$ are then
obtained by fixing $H(x)\eqdef\suc(x)=x+1$.

The Hardy functions enjoy a number of properties; see
\citep{fairtlough92,cichon98}.  They are \emph{expansive}, and
\emph{monotonic} with respect to both
the base function $h$ and to the argument $x$: for all $g\leq h$,
$x\leq y$, and $\alpha$,
\begin{align}\label{eq-exp-mono}
 x&\leq h^\alpha(x)\;,&
 g^\alpha(x)&\leq h^\alpha(x)\;,&
 h^\alpha(x)&\leq h^\alpha(y)\;.
\end{align}
As often with subrecursive functions, what the Hardy functions lack is
monotonicity in the ordinal index, see \autoref{app-mon}.

By transfinite induction on ordinals, we also find several identities:
\begin{align}
  \label{eq-F-h}
  h^{\omega^\alpha\cdot c}&=F^c_{h,\alpha}\;,\\
  \label{eq-comp}
  h^{\alpha+\beta}&=h^\alpha\circ h^\beta\;%
  .
\end{align}
Note that~\eqref{eq-F-h} entails the expansiveness and monotonicity of
the fast-growing functions.

Equation~\eqref{eq-comp} is extremely
valuable: it shows that%
---up to some extent---the \emph{composition}
of Hardy functions can be internalised in
the ordinal index.
Here we run however into a limitation of considering ``set-theoretic''
ordinal indices: informally, \eqref{eq-comp} is %
implicitly restricted to ordinals $\alpha+\beta$ %
``in \ref{CNF}''.  Formally, it requires %
$\alpha+\beta=\alpha\oplus\beta$%
, where ``$\oplus$'' %
denotes the natural sum operation.  For %
instance, it fails in %
$H^{1}(H^\omega(x))=H^{1}(H^x(x+1))=2x+2>2x+1=H^\omega(x)$, although
$1+\omega=\omega$%
.  We will discuss this point further in \autoref{sub-comp}.

\begin{remark}
  Thanks to~\eqref{eq-F-h}, the definitions of the
  $(\FGH{<\alpha})_\alpha$ and $(\F\alpha)_\alpha$ classes can be
  restated purely in terms of the Hardy functions.  Indeed,
  \begin{align*}
    \FGH{<\alpha}&=\bigcup_{\beta<\alpha,c<\omega}\!\CC{FDTime}\big(F^c_\beta(n)\big)\\
    &=\bigcup_{\beta<\alpha,c<\omega}\!\CC{FDTime}\big(H^{\omega^\beta\cdot
      c}(n)\big)\\
    &=\bigcup_{\gamma<\omega^\alpha}\!\CC{FDTime}\big(H^{\gamma}(n)\big)\;,\\
    \F\alpha&=\bigcup_{p\in\FGH{<\alpha}}\CC{DTime}\left(H^{\omega^\alpha}(p(n))\right)\;.
  \end{align*}
\end{remark}

\subsection{Monotonicity}\label{app-mon}
One of the issues of most subrecursive hierarchies of functions is
that they are not monotone in the ordinal index: $\beta<\alpha$ does
not necessarily imply $H^{\beta}\leq H^\alpha$; for instance,
$H^{x+2}(x)=2x+2>2x+1=H^\omega(x)$.  What is true however is that they
are \emph{eventually} monotone: if $\beta<\alpha$, then there exists
$n_0$ such that, for all $x\geq n_0$, $H^\beta(x)\leq H^\alpha(x)$.
This result (and others) can be proven using a \emph{pointwise
  ordering}: for all $x$, define the $\prec_x$ relation as the
transitive closure of
\begin{align}
  \alpha&\prec_x\alpha+1\;, &\lambda(x)\prec_x\lambda\;.
\end{align}
The relation ``$\beta\prec_x\alpha$'' is also noted
``$\beta\in\alpha[x]$'' in \citep[pp.~158--163]{sw12}, where the
results of this section are proven.

The $\prec_x$ relations form a strict hierarchy of refinements of the
ordinal ordering $<$:
\begin{equation}
  {\prec_0}\subsetneq{\prec_1}\subsetneq \cdots\subsetneq
  {\prec_x}\subsetneq\cdots\subsetneq {<}\;.
\end{equation}
We are going to use two main properties of the pointwise ordering:
\begin{align}\label{eq-pointw-lim}
  x&<y&\text{ implies }&&\lambda(x)&\prec_y\lambda(y)\;,\\
  \label{eq-pointwise-mon}
  \beta&\prec_x\alpha&\text{ implies }&&H^\beta(x)&\leq H^\alpha(x)\;.
\end{align}

For a first application, define the \emph{norm} of an ordinal term as
the maximal coefficient that appears in its normal form: if
$\alpha=\omega^{\alpha_1}\cdot c_1+\cdots+\omega^{\alpha_m}\cdot c_m$
with $\alpha_1>\cdots>\alpha_m$ and $c_1,\dots,c_m>0$, then
$N\alpha\eqdef\max\{c_1,\dots,c_m,N\alpha_1,\dots,N\alpha_m\}$.  Then
$\beta<\alpha$ implies
$\beta\prec_{N\beta}\alpha$~\citep[p.~158]{sw12}.  Together with
\eqref{eq-pointwise-mon}, this entails that, for all $x\geq N\beta$,
$H^\beta(x)\leq H^\alpha(x)$.

\subsection{Ackermann Functions}\label{app-ack}
We prove in this section some basic properties of the Ackermann
hierarchy of functions $(A_\alpha)_\alpha$ defined in
\autoref{sub-ack}.  Its definition is less uniform than the
fast-growing and Hardy functions, leading to slightly more involved
proofs.
\begin{lemma}\label{lem-ack-zero}
  For all $\alpha>0$, $A_\alpha(0)\leq 1$.
\end{lemma}
\begin{proof}
  By transfinite induction over $\alpha$.  For $\alpha=1$,
  $A_1(0)=0\leq 1$.  For a successor ordinal $\alpha+1$,
  $A_{\alpha+1}(0)=1$.  For a limit ordinal $\lambda$,
  $A_\lambda(0)=A_{\lambda(0)}(0)\leq 1$ by ind.\ hyp.
\end{proof}

As usual with subrecursive hierarchies, the main issue with the
Ackermann functions is to prove various monotonicity properties in the
argument and in the index.
\begin{lemma}\label{lem-ack}\renewcommand{\theenumi}{\roman{enumi}}
  For all $\alpha,\beta>0$ and $x,y$:
  \begin{enumerate}\renewcommand{\labelenumi}{(\roman{enumi})}
  \item\label{lem-ack-exp}if $\alpha>1$, $A_\alpha$ is strictly
    expansive: $A_\alpha(x)>x$,
  \item\label{lem-ack-mon}$A_\alpha$ is strictly monotone in its
    argument: if $y>x$, $A_\alpha(y)>A_\alpha(x)$,
  \item\label{lem-ack-px}$(A_\alpha)_\alpha$ is pointwise monotone in its
    index: if $\alpha\succ_x\beta$, $A_{\alpha}(x)\geq A_{\beta}(x)$.
  \end{enumerate}
\end{lemma}
\begin{proof}
  Let us first consider the case $\alpha=1$: $A_1$ is
  strictly monotone, proving \eqref{lem-ack-mon}.  Regarding
  \eqref{lem-ack-exp} for $\alpha=2$, $A_2(x)=2^x>x$ for all $x$.

  We prove now the three statements by simultaneous transfinite
  induction over $\alpha$.  Assume they hold for all $\beta<\alpha$
  (and thus for all $\beta\prec_x\alpha$ for all $x$).

  For~\eqref{lem-ack-exp},
  \begin{itemize}
  \item if $\alpha$ is a successor ordinal $\beta+1$,
    then $A_{\beta+1}(x)\geq A_\beta(x)>x$ by
    ind.\ hyp.~\eqref{lem-ack-px} and~\eqref{lem-ack-exp} on
    $\beta\prec_x\alpha$. 
  \item If $\alpha$ is a limit ordinal $\lambda$, then
    $A_\lambda(x)=A_{\lambda(x)}(x)>x$ by
    ind.\ hyp.~\eqref{lem-ack-exp} on $\lambda(x)\prec_x\alpha$.
  \end{itemize}

  For~\eqref{lem-ack-mon}, it suffices to prove the result for $y=x+1$.
  \begin{itemize}
  \item If $\alpha$ is a successor ordinal $\beta+1$, then
    $A_\alpha(x+1)=A_\beta\big(A_\alpha(x)\big)> A_\alpha(x)$ by
    ind.\ hyp.~\eqref{lem-ack-exp} on $\beta\prec_x\alpha$.
  \item If $\alpha$ is a limit ordinal $\lambda$, then
    $A_\lambda(x+1)=A_{\lambda(x+1)}(x+1)\geq A_{\lambda(x)}(x+1)$ by
    ind.\ hyp.~\eqref{lem-ack-px} on
    $\lambda(x)\prec_{x+1}\lambda(x+1)$ (recall
    \autoref{eq-pointw-lim}), hence the result by ind.\
    hyp.~\eqref{lem-ack-mon} on $\lambda_x\prec_x\alpha$.
  \end{itemize}

  For \eqref{lem-ack-px}, it suffices to prove the result for
  $\alpha=\beta+1$ and $\beta=\alpha(x)$ and rely on transitivity.
  \begin{itemize}
  \item If $\alpha=\beta+1$, then we show \eqref{lem-ack-px} by
    induction over $x$: the base case $x=0$ stems from
    $A_\alpha(0)=A^0_\beta(1)=1\geq A_\beta(0)$ by
    \autoref{lem-ack-zero}; the induction step $x+1$ stems from
    $A_\alpha(x+1)=A_\beta\big(A_\alpha(x)\big)\geq A_\beta(x+1)$
    using the ind.\ hyp.~on $x$ and~\eqref{lem-ack-mon} on
    $\beta\prec_{A_\alpha(x)}\alpha$.
  \item If $\beta=\alpha(x)$, then $A_\alpha(x)=A_\beta(x)$ by
    definition.\qedhere
  \end{itemize}
\end{proof}

Our main interest in the Ackermann functions is their relation with
the fast-growing ones:
\begin{lemma}\label{prop-ack}
  For all $\alpha>0$ and all $x$, $A_\alpha(x)\leq F_\alpha(x)\leq
  A_\alpha(6x+5)$.
\end{lemma}
\begin{proof}
  We only prove the second inequality, as the first one can be deduced
  from the various monotonicity properties of $F_\alpha$ and
  $A_\alpha$.  The case $x=0$ is settled for all $\alpha>0$ by
  checking that $F_\alpha(0)=1\leq 10=A_1(5)\leq A_\alpha(5)$, since
  $1\preceq_x\alpha$ for all $\alpha>0$ and we can therefore apply
  \autoref{lem-ack}.\eqref{lem-ack-px}.  Assume now $x>0$; we prove
  the statement by transfinite induction over $\alpha>0$.
  \begin{itemize}
  \item For the base case $\alpha=1$, $F_1(x)=2x+1\leq
    12x+10=A_1(6x+5)$.
  \item For the successor case $\alpha+1$,
    $A_{\alpha+1}(6x+5)=A_{\alpha}^{5(x+1)}\big(A_\alpha^x(1)\big)\geq
    A_\alpha^{5(x+1)}(x)$ by \autoref{lem-ack}.

    We show by induction over $j$ that $A_\alpha^{5j}(x)\geq
    F^j_\alpha(x)$.  This holds for the base case $j=0$, and for the
    induction step, $A_\alpha^5\big(A_\alpha^{5j}(x)\big)\geq
    A_\alpha^5\big(F^j_\alpha(x)\big)$ by ind.\ hyp.\ on $j$ and
    \autoref{lem-ack}.\eqref{lem-ack-mon}.  Furthermore, for all
    $y>0$, $A_\alpha\big(A_\alpha^4(y)\big)\geq
    A_\alpha\big(A_1^4(y)\big)=A_\alpha(16y)\geq A_\alpha(6y+5)\geq
    F_\alpha(y)$ by ind.\ hyp.\ on $\alpha$, which shows that
    $A_\alpha^5\big(F^j_\alpha(x)\big)\geq F_\alpha^{j+1}(x)$ when
    choosing $y=F^j_\alpha(x)>0$.  Then $A_\alpha^{5(x+1)}(x)\geq
    F_\alpha^{x+1}(x)=F_{\alpha+1}(x)$, thus completing the proof in
    the successor case.
  \item For the limit case $\lambda$,
    $A_\lambda(6x+5)=A_{\lambda(6x+5)}(6x+5)\geq
    A_{\lambda(x)}(6x+5)\geq F_{\lambda(x)}(x)=F_\lambda(x)$, using
    successively \autoref{lem-ack}.\eqref{lem-ack-px} on
    $\lambda(x)\prec_{6x+5}\lambda(6x+5)$ and the ind.\ hyp.\ on
    $\lambda(x)<\lambda$.\qedhere
  \end{itemize}
\end{proof}

\subsection{Relativised Functions}\label{app-rel}We prove here the
missing lemma from the proof of \autoref{th-rel}:

\begin{lemma}\label{lem-rel}\renewcommand{\theenumi}{\roman{enumi}}
  Let $h{:}\,\+N\to\+N$ be a function, $\alpha,\beta$ be two ordinals,
  and $x_0$ be a natural number.  If for all $x\geq x_0$, $h(x)\leq
  F_\beta(x)$, then there exists an ordinal $\gamma$ such that
  \begin{enumerate}\renewcommand{\labelenumi}{(\roman{enumi})}
  \item\label{lem-rel-1} for all $x\geq x_0$, $F_{h,\alpha}(x)\leq
    F_{\beta+\alpha}(F_{\gamma}(x))$, and
  \item\label{lem-rel-2} $\gamma<\beta+\alpha$ whenever $\beta+\alpha>0$.
  \end{enumerate}
\end{lemma}
\begin{proof}
  Let us first fix some notations: write
  $\alpha=\omega^{\alpha_1}+\cdots+\omega^{\alpha_m}$ with
  $\alpha_1\geq\cdots\geq\alpha_m$ and
  $\beta=\omega^{\beta_1}+\cdots+\omega^{\beta_n}$ with
  $\beta_1\geq\cdots\geq\beta_n$, and let $i$ be the maximal index in
  $\{1,\dots,n\}$ such that $\beta_i\geq\alpha_1$, or set $i=0$ if
  this does not occur.  Define
  $\beta'\eqdef\omega^{\beta_1}+\cdots+\omega^{\beta_i}$ and
  $\gamma\eqdef\omega^{\beta_{i+1}}+\cdots+\omega^{\beta_n}$ (thus
  $\beta'=0$ if $i=0$); then $\beta=\beta'+\gamma$ and
  $\beta+\alpha=\beta'+\alpha$.  Note that this implies
  $\gamma<\omega^{\alpha_1}\leq\alpha\leq\beta+\alpha$, unless
  $\alpha=0$ and then $\gamma=0$, thus fulfilling \eqref{lem-rel-2}.

  We first prove by transfinite induction over $\alpha$ that
  \begin{equation}\label{eq-rel-proof}
    F_{\beta'+\alpha}\circ F_{\gamma}\geq F_\gamma\circ F_{F_\beta,\alpha}\;.
  \end{equation}
  \begin{proof}[Proof of \eqref{eq-rel-proof}]
    For the base case $\alpha=0$, then $\gamma=0$ and $\beta'=\beta$,
    and indeed
    \begin{align*}
      F_\beta(F_0(x))&=F_\beta(x+1)\\
      &\geq F_\beta(x)+1&&\text{by monotonicity of $F_\beta$}\\
      &=F_0(F_\beta(x))\\&=F_0(F_{F_\beta,0}(x))\;.
    \end{align*}
    For the successor case $\alpha+1$ and assuming it holds for
    $\alpha$, let us first show by induction over $j$ that, for all $y$,
    \begin{equation}\label{eq-rel-j}
      F^j_{\beta'+\alpha}(F_\gamma(y))\geq
      F_\gamma(F^j_{F_\beta,\alpha}(y))\;.
    \end{equation}
    This immediately holds for the base case $j=0$, and for the
    induction step,
    \begin{align*}
      F_{\beta'+\alpha}\big(F^j_{\beta'+\alpha}(F_\gamma(y))\big)&\geq
      F_{\beta'+\alpha}\big(F_\gamma(F^j_{F_\beta,\alpha}(y))\big)&&\text{by
        ind.\ hyp.~\eqref{eq-rel-j} on $j$}\\
      &\geq
      F_\gamma\big(F_{F_\beta,\alpha}(F^j_{F_\beta,\alpha}(y))\big)&&\text{by
        ind.\ hyp.~\eqref{eq-rel-proof} on $\alpha<\alpha+1$.}
    \end{align*}
    This yields the desired inequality:
    \begin{align*}
      F_{\beta'+\alpha+1}(F_\gamma(x))&=F_{\beta'+\alpha}^{F_\gamma(x)+1}(F_\gamma(x))\\
      &\geq F_{\beta'+\alpha}^{x+1}(F_\gamma(x))\\
      &\geq F_\gamma(F_{F_\beta,\alpha}^{x+1}(x))\\&=F_\gamma(F_{F_\beta,\alpha+1}(x))
    \end{align*}
    using \eqref{eq-rel-j} with $j=x+1$ and $y=x$.
    
    For the limit case $\lambda$,
    \begin{align*}
      F_{\beta'+\lambda}(F_\gamma(x))&=F_{\beta'+\lambda(F_\gamma(x))}(F_\gamma(x))\\
      &\leq F_{\beta'+\lambda(x)}(F_\gamma(x))&&\text{since
        $\lambda(x)\prec_{F_\gamma(x)}\lambda(F_\gamma(x))$}\\
      &\leq F_\gamma(F_{F_\beta,\lambda(x)}(x))&&\text{by ind.\
        hyp.~\eqref{eq-rel-proof} on $\lambda(x)<\lambda$}\\
      &= F_\gamma(F_{F_\beta,\lambda}(x))\;.&&\qedhere
    \end{align*}
  \end{proof}

  Returning to the main proof, a simple induction over $\alpha$ shows
  that, for all $x\geq x_0$,
  \begin{equation}\label{eq-rel-h}
    F_{h,\alpha}(x)\leq F_{F_\beta,\alpha}(x)\;.
  \end{equation}
  We then conclude for \eqref{lem-rel-1} that, for all $x\geq x_0$,
  \begin{align*}
    F_{h,\alpha}(x)&\leq F_{F_\beta,\alpha}(x)&&\text{by \eqref{eq-rel-h}}\\
    &\leq F_\gamma(F_{F_\beta,\alpha}(x))&&\text{by expansivity of
      $F_\gamma$}\\
    &\leq F_{\beta'+\alpha}(F_\gamma(x))&&\text{by \eqref{eq-rel-proof}.}\qedhere
  \end{align*}
\end{proof}

\subsection{Non-standard Assignment of Fundamental Sequences}
\label{ap--fund}We show here the omitted details of the proof of
\autoref{th-fund}:
\begin{lemma}\label{lem-fund}
  Let $s{:}\+N\to\+N$ be a monotone function and $\alpha$ be an
  ordinal.
  \begin{itemize}
    \item If $s$ is strictly expansive, then $F_{\alpha,s}\leq
      F_{s,\alpha}\circ s$, and
    \item otherwise $F_{\alpha,s}\leq F_{\alpha}$.
  \end{itemize}
\end{lemma}
\begin{proof}
  For the first point, let us show that
  \begin{equation}
    s(F_{\alpha,s}(x))\leq F_{s,\alpha}(s(x))
  \end{equation}
  for all monotone $s$ with $s(x)>x$, all $\alpha$ and all $x$, which
  entails the lemma since $s$ is expansive.  We
  proceed by transfinite induction over $\alpha$.  For the base case,
  $F_{s,0}(s(x))=s(s(x))\geq s(x+1)=s(F_{0,s}(x))$ since $s$ is
  monotone and strictly expansive.  For the successor case,
  $F_{s,\alpha+1}(s(x))=F_{s,\alpha}^{s(x)+1}(s(x))\geq
  s(F_{\alpha,s}^{s(x)}(x))=s(F_{\alpha+1,s}(x))$, where the middle
  inequality stems from the fact that $F_{s,\alpha}^j(s(x))\geq
  s(F_{\alpha,s}^j(x))$, as can be seen by induction on $j$ using the
  induction hypothesis on $\alpha<\alpha+1$.  For the
  limit case, observe that $\lambda(x)_s\prec_{s(x)}\lambda(s(x))$,
  thus $F_{s,\lambda}(s(x))=F_{s,\lambda(s(x))}(s(x))\geq
  F_{s,\lambda(x)_s}(s(x))\geq
  s(F_{\lambda(x)_s,s}(x))=s(F_{\lambda,s}(x))$ using the induction
  hypothesis on $\lambda(x)_s<\lambda$.

  The second point is straightforward by induction over $\alpha$.
\end{proof}

\begin{lemma}\label{lem-id-fund}
  For all $\alpha$, $F_0\circ F_{\alpha}\leq F_{\alpha,\mathrm{id}}\circ F_0$.
\end{lemma}
\begin{proof}
  By induction over $\alpha$.  For the zero case,
  $F_0(F_0(x))=x+2=F_{0,\mathrm{id}}(F_0(x))$.  For the successor
  case, we can check that $F^j_{\alpha,\mathrm{id}}(x+1)\geq
  F^j_{\alpha}(x)+1$ for all $j$ using the induction hypothesis on
  $\alpha$, thus
  $F_{\alpha,\mathrm{id}}(x+1)=F^{x+1}_{\alpha,\mathrm{id}}(x+1)\geq
  F^{x+1}_\alpha(x)+1=F_{\alpha+1}(x)+1$.  For the limit case, note
  that $\lambda(x)\prec_{x+1}\lambda(x+1)$ thus
  $F_{\lambda,\mathrm{id}}(x+1)=F_{\lambda_{x+1},\mathrm{id}}(x+1)\geq
  F_{\lambda(x+1)}(x)+1\geq F_{\lambda(x)}(x)+1=F_\lambda(x)+1$.
\end{proof}

\subsection{Composing Hardy Functions}\label{sub-comp}The purpose of
this section is to provide the technical details for the proof of
\autoref{th-redc}.

The \emph{natural sum} $\alpha\oplus\beta$ of two ordinals written as
$\alpha=\omega^{\alpha_1}+\cdots+\omega^{\alpha_m}$ with
$\alpha_1\geq\cdots\geq\alpha_m$ and
$\beta=\omega^{\beta_1}+\cdots\omega^{\beta_n}$ with
$\beta_1\geq\cdots\geq\beta_n$ can be defined as the ordinal
$\omega^{\gamma_1}+\cdots+\omega^{\gamma_{m+n}}$ where the
$\gamma_i$'s range over $\{\alpha_j\mid 1\leq j\leq
m\}\cup\{\beta_k\mid 1\leq k\leq n\}$ in non-increasing order.  For
instance, $\omega^2+\omega^{\omega}=\omega^\omega$ but
$\omega^2\oplus\omega^\omega=\omega^\omega+\omega^2$.

\begin{lemma}\label{prop-comp}
  For all ordinals $\alpha$ and $\beta$, and all functions $h$,
  $$h^\alpha\circ h^\beta\leq h^{\alpha\oplus\beta}\;.$$
\end{lemma}
\begin{proof}
  Write $\alpha=\omega^{\alpha_1}+\cdots+\omega^{\alpha_m}$ with
  $\alpha_1\geq\cdots\geq\alpha_m$ and
  $\beta=\omega^{\beta_1}+\cdots+\omega^{\beta_n}$ with
  $\beta_1\geq\cdots\geq\beta_n$, then
  $\alpha\oplus\beta=\omega^{\gamma_1}+\cdots+\omega^{\gamma_{m+n}}$.
  We prove the lemma by transfinite induction over $\beta$: it holds
  immediately for the base case since $\alpha\oplus 0=\alpha$ and for
  the successor case since
  $\alpha\oplus(\beta+1)=(\alpha\oplus\beta)+1$.  For the limit case,
  let $i$ be the last index of $\beta_n$ among the $\gamma_j$ in the
  \ref{CNF} of $\alpha\oplus\beta$.  If $i=m+n$, then
  $\alpha\oplus(\beta(x))=(\alpha\oplus\beta)(x)$ and the statement
  holds.  Otherwise, define
  $\gamma\eqdef\omega^{\gamma_1}+\cdots+\omega^{\gamma_i}$ and
  $\gamma'\eqdef\omega^{\gamma_{i+1}}+\cdots+\omega^{\gamma_{m+n}}$.
  For all $x$,
  \begin{align*}
    h^{\alpha\oplus\beta}&=h^{\gamma}(h^{\gamma'}(x))&&\text{by \eqref{eq-comp}}\\
    &= h^{\gamma(h^{\gamma'}(x))}(h^{\gamma'}(x))&&\text{since
      $\gamma$ is a limit ordinal}\\
    &\geq h^{\gamma(x)}(h^{\gamma'}(x))&&\text{since $\gamma(x)\prec_{[h^{\gamma'}\!(x)]}\gamma(h^{\gamma'}(x))$}\\
    &= h^{\alpha\oplus(\beta(x))}(x)&&\text{by \eqref{eq-comp}}\\
    &\geq h^\alpha(h^{\beta(x)}(x))&&\text{by ind.\ hyp.\ on $\beta(x)<\beta$}\\
    &= h^\alpha(h^\beta(x))\;.&&\qedhere
  \end{align*}
\end{proof}

\begin{corollary}\label{cor-comp}
  Let $\alpha$ be an ordinal and $f$ a function in $\FGH{<\alpha}$.
  Then there exists $g$ in $\FGH{<\alpha}$ such that $f\circ
  F_\alpha\leq F_\alpha\circ g$.
\end{corollary}
\begin{proof}
  As $f$ is in some $\FGH\beta$ for $\beta<\alpha$, $f\leq F_\beta^c$
  for some finite $c$ by \citep[\theoremautorefname~2.10]{lob70}, thus
  $f\leq H^{\omega^\beta\cdot c}$ by \eqref{eq-F-h}, and we let
  $g\eqdef H^{\omega^\beta\cdot c}$, which indeed belongs to
  $\FGH\beta\subseteq\FGH{<\alpha}$.  Still by \eqref{eq-F-h},
  $F_\alpha=H^{\omega^\alpha}$.  Observe that $\omega^\beta\cdot
  c<\omega^\alpha$, hence $(\omega^\beta\cdot
  c)\oplus\omega^\alpha=\omega^\alpha+\omega^\beta\cdot c$.  By
  \eqref{eq-comp}, $H^{\omega^\alpha+\omega^\beta\cdot
    c}=H^{\omega^\alpha}\circ H^{\omega^\beta\cdot c}$. Applying
  \eqref{eq-F-h} and \autoref{prop-comp}, we obtain that $f\circ
  F_\alpha\leq g\circ F_\alpha\leq F_\alpha\circ g$.
\end{proof}

\subsection{Computing Hardy Functions}\label{sub-honest}
We explain in this section how to compute Hardy functions, thus
providing the background material for the proof of
\autoref{th-constr}.  This type of results is pretty standard---see
for instance \citep{wainer70}, \citep{fairtlough98}, or
\citep[pp.~159--160]{sw12}---, but the particular way we employ is
closer in spirit to the viewpoint employed in
\citep{HSS-lics2012,fossacs/KarandikarS13,HaaseSS13}.

\subsubsection{Hardy Computations.}
Using \eqref{eq-hardy-pred}, let us call a \emph{Hardy computation}
for $h^\alpha(n)$ a sequence of pairs
$\tup{\alpha_0,n_0},\tup{\alpha_1,n_1},\dots,\tup{\alpha_\ell,n_\ell}$
where $\alpha_0=\alpha$, $n_0=n$, $\alpha_\ell=0$, and at each step
$0<i\leq\ell$, $\alpha_i=P_{n_{i-1}}(\alpha_{i-1})$ and
$n_i=h(n_{i-1})$.  An invariant of this computation is that
$h^{\alpha_i}(n_i)=h^\alpha(n)$ at all steps $0\leq i\leq\ell$, hence
$n_\ell=h^\alpha(n)$.  Since $h$ is increasing, the $n_i$ values
increase throughout this computation, while the $\alpha_i$ values
decrease, and termination is guaranteed.

Our plan is to implement the Hardy computation of $h^\alpha(n)$ using
a Turing machine, which essentially needs to implement the $\ell$
steps
$\tup{\alpha_i,n_i}\to\tup{P_{n_{i-1}}(\alpha_{i-1}),h(n_{i-1})}$.  We
assume $h$ to be an elementarily constructible expansive function,
such that $h(n)$ can be computed in $e(h(n))$ for some fixed monotone
elementary function $e$.  Then, the complexity of a single step will
depend mainly on $h(n_{i-1})\leq h^\ell(n)$ and on the complexity of
updating $\alpha_i$.

\subsubsection{Cicho\'n Functions.}In order to measure the length
$\ell$ of a Hardy computation for $h^\alpha(n)$, we define a family
$(h_\alpha)_\alpha$ of functions $\+N\to\+N$ by induction on the
ordinal index:
\begin{align}
  h_0(x)&\eqdef 0\;,&
  h_{\alpha+1}(x)&\eqdef 1+h_\alpha(h(x))\;,&
  h_{\lambda}(x)&\eqdef h_{\lambda(x)}(x)\;.
\end{align}
This family is also known as the \emph{length hierarchy} and was
defined by \citet{cichon98}.  It satisfies several interesting
identities:
\begin{align}
  h^\alpha(x)&=h^{h_\alpha(x)}(x)\;,&
  h^\alpha(x)&\geq h_\alpha(x)+x\;.
\end{align}
Its main interest here is that it measures the length of Hardy
computations: $\ell=h_\alpha(n)\leq h^\alpha(n)$ by the above
equations, which in turn implies $h^\ell(n)=h^\alpha(n)$.

\subsubsection{Encoding Ordinal Terms.} It remains to bound the
complexity of computing $\alpha_i=P_{n_{i-1}}(\alpha_{i-1})$.
Assuming some reasonable string encoding of the terms denoting the
$\alpha_{i}$~\citep[e.g.][]{HaaseSS13}, we will consider that each
$\alpha_i$ can be computed in time $p(|\alpha_i|)$ a monotone
polynomial function of the size $|\alpha_i|$ of its term
representation, and will rather concentrate on bounding this size.  We
define it by induction on the term denoting $\alpha_i$:
\begin{align}
  |0|&\eqdef 0\;,&
  |\omega^\alpha|&\eqdef 1+|\alpha|\;,&
  |\alpha+\alpha'|&\eqdef |\alpha|+|\alpha'|\;.
\end{align}

Let us also recall the definition of the \emph{slow-growing hierarchy}
$(G_\alpha)_\alpha$:
\begin{align}
  G_0(x)&\eqdef 0\;,&
  G_{\alpha+1}(x)&\eqdef 1+G_\alpha(x)\;,&
  G_\lambda(x)&\eqdef G_{\lambda(x)}(x)\;.
\end{align}
The slow-growing function satisfy several natural identities
\begin{align}
  \label{eq-pred-G}G_{\alpha}(x)&=1+G_{P_x(\alpha)}(x)\;,\\
  \label{eq-mon-G}G_{\alpha}(x+1)&>G_\alpha(x)\;,\\
  \label{eq-pw-G}\text{if }\beta\prec_x\alpha\text{ then } G_{\beta}(x)&\leq
  G_\alpha(x)\;.
\end{align}
  Furthermore,
\begin{align}
  G_{\alpha+\alpha'}(x)&=G_{\alpha}(x)+G_{\alpha'}(x)\;,&
  G_{\omega^\alpha}(x)&=(x+1)^{G_\alpha(x)}\;.
\end{align}
Hence, $G_\alpha(x)$ is the elementary function which results from
substituting $x+1$ for every occurrence of $\omega$ in the Cantor
normal form of $\alpha$~\citep[p.~159]{sw12}.

\begin{lemma}\label{lem-norm-G}
  Let $x>0$.  Then $|\alpha|\leq G_\alpha(x)$.
\end{lemma}
\begin{proof}
  By induction over the term denoting $\alpha$: $|0|=0=G_0(x)$,
  $|\omega^\alpha|=1+|\alpha|\leq (x+1)^{|\alpha|}\leq
  (x+1)^{G_\alpha(x)}=G_{\omega^\alpha}(x)$, and
  $|\alpha+\alpha'|=|\alpha|+|\alpha'|\leq
  G_{\alpha}(x)+G_{\alpha'}(x)=G_{\alpha+\alpha'}(x)$.
\end{proof}

\begin{lemma}\label{lem-term-G}
  If $\tup{\alpha_0,n_0},\dots,\tup{\alpha_\ell,n_\ell}$ is a Hardy
  computation for $h^\alpha(n)$ with $n>0$, then for all $0\leq
  i\leq\ell$, $|\alpha_i|\leq G_{\alpha}(n_\ell)$.
\end{lemma}
\begin{proof}
  We distinguish two cases.  If $i=0$, then
  $|\alpha_0|=|\alpha|\leq G_\alpha(n)$ by \autoref{lem-norm-G} since
  $n>0$, hence $|\alpha_0|\leq G_\alpha(n_\ell)$ since $n_\ell\geq n$
  by \eqref{eq-mon-G}.
  If $i>0$, then
  \begin{align*}
    |\alpha_i|&=|P_{n_{i-1}}(\alpha_{i-1})|\\
    &\leq G_{P_{n_{i-1}}(\alpha_{i-1})}(n_{i-1})&&\text{by
      \autoref{lem-norm-G} since $n_{i-1}\geq n>0$}\\
    &<G_{\alpha_{i-1}}(n_{i-1})&&\text{by \eqref{eq-pred-G}}\\
    &\leq G_\alpha(n_{i-1})&&\text{since
      $\alpha_{i-1}\prec_{n_{i-1}}\alpha$ by \eqref{eq-pw-G}}\\
    &\leq G_\alpha(n_\ell)&&\text{since $n_{i-1}\leq n_\ell$ by \eqref{eq-mon-G}}\qedhere
  \end{align*}
\end{proof}
The restriction to $n>0$ in \autoref{lem-term-G} is not a big issue:
either $h(0)=0$ and then $h^\alpha(0)=0$, or $h(0)>0$ and then
$h^{\gamma+\omega^\beta}(0)=h^\gamma(h(0))$ and we can proceed from
$\gamma$ instead of $\gamma+\omega^\beta$ as initial ordinal of our
computation.

\subsubsection{Wrapping up.}
To conclude, each of the $\ell\leq h^\alpha(n)$ steps of a Hardy
computation for $h^\alpha(n)$ needs to compute
\begin{itemize}
\item $\alpha_i$, in time $p(G_\alpha(h^\alpha(n)))$ since
  $|\alpha_i|\leq G_\alpha(h^\alpha(n))$ and $p$ was assumed monotone,
  and
\item $n_i$, in time $e(h^\alpha(n))$ since $h(n_{i-1})\leq
  h^\alpha(n)$ and $e$ was assumed monotone.
\end{itemize}
This yields the following statement:
\begin{proposition}\label{prop-constr}
  The Hardy function $h^\alpha$ can be computed in time
  \begin{equation*}O\!\!\left(h^{\alpha}(n)\cdot\big(p(G_{\alpha}(h^\alpha(n)))+e(h^\alpha(n))\big)\right)\;.\end{equation*}
\end{proposition}

%% file: amsart.bbl
\begin{thebibliography}{94}
\providecommand{\natexlab}[1]{#1}
\providecommand{\url}[1]{\texttt{#1}}
\expandafter\ifx\csname urlstyle\endcsname\relax
  \providecommand{\doi}[1]{doi: #1}\else
  \providecommand{\doi}{doi: \begingroup \urlstyle{rm}\Url}\fi

\bibitem[Abdulla and Delzanno(2006)]{abdulla06}
P.~A. Abdulla and G.~Delzanno.
\newblock On the coverability problem for constrained multiset rewriting.
\newblock In \emph{AVIS 2006}, 2006.

\bibitem[Abdulla and Jonsson(1996)]{abdulla96b}
P.~A. Abdulla and B.~Jonsson.
\newblock Verifying programs with unreliable channels.
\newblock \emph{Inform. and Comput.}, 127\penalty0 (2):\penalty0 91--101, 1996.
\newblock \doi{10.1006/inco.1996.0053}.

\bibitem[Abdulla and Nyl\'en(2001)]{abdulla01}
P.~A. Abdulla and A.~Nyl\'en.
\newblock Timed {P}etri nets and {BQO}s.
\newblock In \emph{Petri Nets 2001}, volume 2075 of \emph{Lect. Notes in
  Comput. Sci.}, pages 53--70. Springer, 2001.
\newblock \doi{10.1007/3-540-45740-2_5}.

\bibitem[Abdulla et~al.(2000)Abdulla, {\v{C}}er{\=a}ns, Jonsson, and
  Tsay]{abdulla2000c}
P.~A. Abdulla, K.~{\v{C}}er{\=a}ns, B.~Jonsson, and Y.-K. Tsay.
\newblock Algorithmic analysis of programs with well quasi-ordered domains.
\newblock \emph{Inform. and Comput.}, 160\penalty0 (1--2):\penalty0 109--127,
  2000.
\newblock \doi{10.1006/inco.1999.2843}.

\bibitem[Abdulla et~al.(2011)Abdulla, Delzanno, and Van{ }Begin]{abdulla2011}
P.~A. Abdulla, G.~Delzanno, and L.~Van{ }Begin.
\newblock A classification of the expressive power of well-structured
  transition systems.
\newblock \emph{Inform. and Comput.}, 209\penalty0 (3):\penalty0 248--279,
  2011.
\newblock \doi{10.1016/j.ic.2010.11.003}.

\bibitem[Abriola et~al.(2015)Abriola, Figueira, and Senno]{abriola}
S.~Abriola, S.~Figueira, and G.~Senno.
\newblock Linearizing well-quasi orders and bounding the length of bad
  sequences.
\newblock \emph{Theor. Comput. Sci.}, 603:\penalty0 3--22, 2015.
\newblock \doi{10.1016/j.tcs.2015.07.012}.

\bibitem[Alur and Dill(1994)]{alur94}
R.~Alur and D.~L. Dill.
\newblock A theory of timed automata.
\newblock \emph{Theor. Comput. Sci.}, 126\penalty0 (2):\penalty0 183--235,
  1994.
\newblock \doi{10.1016/0304-3975(94)90010-8}.

\bibitem[Atig et~al.(2010)Atig, Bouajjani, Burckhardt, and Musuvathi]{tsoreach}
M.~F. Atig, A.~Bouajjani, S.~Burckhardt, and M.~Musuvathi.
\newblock On the verification problem for weak memory models.
\newblock In \emph{POPL 2010}, pages 7--18. ACM, 2010.
\newblock \doi{10.1145/1706299.1706303}.

\bibitem[Barcel\'o et~al.(2013)Barcel\'o, Figueira, and Libkin]{barcelo12}
P.~Barcel\'o, D.~Figueira, and L.~Libkin.
\newblock Graph logics with rational relations.
\newblock \emph{Logic. Meth. in Comput. Sci.}, 9\penalty0 (3:1), 2013.
\newblock \doi{10.2168/LMCS-9(3:1)2013}.

\bibitem[Beckmann(2001)]{beckmann01}
A.~Beckmann.
\newblock Exact bounds for lengths of reductions in typed $\lambda$-calculus.
\newblock \emph{J.~Symb. Log.}, 66\penalty0 (3):\penalty0 1277--1285, 2001.
\newblock \doi{10.2307/2695106}.

\bibitem[Blockelet and Schmitz(2011)]{mfcs/BlockeletS11}
M.~Blockelet and S.~Schmitz.
\newblock Model-checking coverability graphs of vector addition systems.
\newblock In \emph{MFCS 2011}, volume 6907 of \emph{Lect. Notes in Comput.
  Sci.}, pages 108--119. Springer, 2011.
\newblock \doi{10.1007/978-3-642-22993-0_13}.

\bibitem[Bonnet et~al.(2010)Bonnet, Finkel, Haddad, and
  Rosa{-}Velardo]{rr-lsv-10-23}
R.~Bonnet, A.~Finkel, S.~Haddad, and F.~Rosa{-}Velardo.
\newblock Comparing {{Petri} Data Nets} and {Timed {Petri} Nets}.
\newblock Research Report LSV-10-23, LSV, ENS Cachan, Dec. 2010.
\newblock URL \url{http://lsv.fr/Publis/rrpublis?onlykey=rr-lsv-10-23}.

\bibitem[Bouyer et~al.(2012)Bouyer, Markey, Ouaknine, Schnoebelen, and
  Worrell]{bouyer12}
P.~Bouyer, N.~Markey, J.~O. Ouaknine, {\relax Ph}.~Schnoebelen, and J.~B.
  Worrell.
\newblock On termination and invariance for faulty channel machines.
\newblock \emph{Form. Asp. Comput.}, 24\penalty0 (4--6):\penalty0 595--607,
  2012.
\newblock \doi{10.1007/s00165-012-0234-7}.

\bibitem[Bresolin et~al.(2012)Bresolin, Della~Monica, Montanari, Sala, and
  Sciavicco]{bresolin2012}
D.~Bresolin, D.~Della~Monica, A.~Montanari, P.~Sala, and G.~Sciavicco.
\newblock Interval temporal logics over finite linear orders: The complete
  picture.
\newblock In \emph{ECAI 2012}, volume 242 of \emph{Frontiers in Artificial
  Intelligence and Applications}, pages 199--204. IOS, 2012.
\newblock \doi{10.3233/978-1-61499-098-7-199}.

\bibitem[C\'ec\'e et~al.(1996)C\'ec\'e, Finkel, and {Purushothaman
  Iyer}]{cece95}
G.~C\'ec\'e, A.~Finkel, and S.~{Purushothaman Iyer}.
\newblock Unreliable channels are easier to verify than perfect channels.
\newblock \emph{Inform. and Comput.}, 124\penalty0 (1):\penalty0 20--31, 1996.
\newblock \doi{10.1006/inco.1996.0003}.

\bibitem[Chambart and Schnoebelen(2007)]{pepreg}
P.~Chambart and {\relax Ph}.~Schnoebelen.
\newblock Post embedding problem is not primitive recursive, with applications
  to channel systems.
\newblock In \emph{FSTTCS 2007}, volume 4855 of \emph{Lect. Notes in Comput.
  Sci.}, pages 265--276. Springer, 2007.
\newblock \doi{10.1007/978-3-540-77050-3_22}.

\bibitem[Chambart and Schnoebelen(2008{\natexlab{a}})]{CS-fossacs08}
P.~Chambart and {\relax Ph}.~Schnoebelen.
\newblock The \(\omega\)-regular {P}ost embedding problem.
\newblock In \emph{FoSSaCS 2008}, volume 4962 of \emph{Lect. Notes in Comput.
  Sci.}, pages 97--111. Springer, 2008{\natexlab{a}}.
\newblock \doi{10.1007/978-3-540-78499-9_8}.

\bibitem[Chambart and Schnoebelen(2008{\natexlab{b}})]{lcs}
P.~Chambart and {\relax Ph}.~Schnoebelen.
\newblock The ordinal recursive complexity of lossy channel systems.
\newblock In \emph{LICS 2008}, pages 205--216. IEEE, 2008{\natexlab{b}}.
\newblock \doi{10.1109/LICS.2008.47}.

\bibitem[Cicho\'n and {Tahhan Bittar}(1998)]{cichon98}
E.~A. Cicho\'n and E.~{Tahhan Bittar}.
\newblock Ordinal recursive bounds for {Higman}'s {T}heorem.
\newblock \emph{Theor. Comput. Sci.}, 201\penalty0 (1--2):\penalty0 63--84,
  1998.
\newblock \doi{10.1016/S0304-3975(97)00009-1}.

\bibitem[Clote(1986)]{clote}
P.~Clote.
\newblock On the finite containment problem for {P}etri nets.
\newblock \emph{Theor. Comput. Sci.}, 43:\penalty0 99--105, 1986.
\newblock \doi{10.1016/0304-3975(86)90169-6}.

\bibitem[Clote(1999)]{clote99}
P.~Clote.
\newblock Computation models and function algebras.
\newblock In \emph{Handbook of Computability Theory}, volume 140 of
  \emph{Studies in Logic and the Foundations of Mathematics}, chapter~17, pages
  589--681. Elsevier, 1999.
\newblock \doi{10.1016/S0049-237X(99)80033-0}.

\bibitem[Dauchet and Tison(1990)]{dauchet90}
M.~Dauchet and S.~Tison.
\newblock The theory of ground rewrite systems is decidable.
\newblock In \emph{LICS~'90}, pages 242--248. IEEE, 1990.
\newblock \doi{10.1109/LICS.1990.113750}.

\bibitem[de~Jongh and Parikh(1977)]{dejongh77}
D.~H.~J. de~Jongh and R.~Parikh.
\newblock Well-partial orderings and hierarchies.
\newblock \emph{Indag. Math.}, 39\penalty0 (3):\penalty0 195--207, 1977.
\newblock \doi{10.1016/1385-7258(77)90067-1}.

\bibitem[Decker and Thoma(2016)]{decker15}
N.~Decker and D.~Thoma.
\newblock On freeze {LTL} with ordered attributes.
\newblock In \emph{FoSSaCS 2016}, Lect. Notes in Comput. Sci. Springer, 2016.
\newblock URL \url{http://arxiv.org/abs/1504.06355}.
\newblock To appear.

\bibitem[Demri and Lazi{\'c}(2009)]{demri09}
S.~Demri and R.~Lazi{\'c}.
\newblock {LTL} with the freeze quantifier and register automata.
\newblock \emph{ACM Trans. Comput. Logic}, 10\penalty0 (3):\penalty0
  16:1--16:30, 2009.
\newblock \doi{10.1145/1507244.1507246}.

\bibitem[Dunn and Restall(2002)]{dunn02}
J.~M. Dunn and G.~Restall.
\newblock Relevance logic.
\newblock In \emph{Handbook of Philosophical Logic}, volume~6, pages 1--128.
  Kluwer, 2002.
\newblock \doi{10.1007/978-94-017-0460-1_1}.

\bibitem[Elgaard et~al.(1998)Elgaard, Klarlund, and M{\o}ller]{mona}
J.~Elgaard, N.~Klarlund, and A.~M{\o}ller.
\newblock {MONA} 1.x\string: new techniques for {WS1S} and {WS2S}.
\newblock In \emph{CAV~'98}, volume 1427 of \emph{Lect. Notes in Comput. Sci.},
  pages 516--520. Springer, 1998.
\newblock \doi{10.1007/BFb0028773}.

\bibitem[Fairtlough and Wainer(1992)]{fairtlough92}
M.~V.~H. Fairtlough and S.~S. Wainer.
\newblock Ordinal complexity of recursive definitions.
\newblock \emph{Inform. and Comput.}, 99\penalty0 (2):\penalty0 123--153, 1992.
\newblock \doi{10.1016/0890-5401(92)90027-D}.

\bibitem[Fairtlough and Wainer(1998)]{fairtlough98}
M.~V.~H. Fairtlough and S.~S. Wainer.
\newblock Hierarchies of provably recursive functions.
\newblock In \emph{Handbook of Proof Theory}, volume 137 of \emph{Studies in
  Logic and the Foundations of Mathematics}, chapter III, pages 149--207.
  Elsevier, 1998.
\newblock \doi{10.1016/S0049-237X(98)80018-9}.

\bibitem[Feferman(1962)]{feferman62}
S.~Feferman.
\newblock Classification of recursive functions by means of hierarchies.
\newblock \emph{Trans. Amer. Math. Soc.}, 104:\penalty0 101--122, 1962.
\newblock \doi{10.1090/S0002-9947-1962-0142453-3}.

\bibitem[Figueira(2012)]{figueira12}
D.~Figueira.
\newblock Alternating register automata on finite words and trees.
\newblock \emph{Logic. Meth. in Comput. Sci.}, 8\penalty0 (1:22), 2012.
\newblock \doi{10.2168/LMCS-8(1:22)2012}.

\bibitem[Figueira and Segoufin(2009)]{FigSeg-mfcs09}
D.~Figueira and L.~Segoufin.
\newblock Future-looking logics on data words and trees.
\newblock In \emph{MFCS 2009}, volume 5734 of \emph{Lect. Notes in Comput.
  Sci.}, pages 331--343. Springer, 2009.
\newblock \doi{10.1007/978-3-642-03816-7_29}.

\bibitem[Figueira et~al.(2011)Figueira, Figueira, Schmitz, and
  Schnoebelen]{FFSS-lics2011}
D.~Figueira, S.~Figueira, S.~Schmitz, and {\relax Ph}.~Schnoebelen.
\newblock {A}ckermannian and primitive-recursive bounds with {D}ickson's
  {L}emma.
\newblock In \emph{LICS 2011}, pages 269--278. IEEE, 2011.
\newblock \doi{10.1109/LICS.2011.39}.

\bibitem[Figueira et~al.(2015)Figueira, Hofman, and Lasota]{FHL10}
D.~Figueira, P.~Hofman, and S.~Lasota.
\newblock Relating timed and register automata.
\newblock \emph{Math. Struct. Comput. Sci.}, 2015.
\newblock \doi{10.1017/S0960129514000322}.
\newblock To appear.

\bibitem[Finkel(1987)]{finkel87c}
A.~Finkel.
\newblock A generalization of the procedure of {Karp} and {Miller} to well
  structured transition systems.
\newblock In \emph{ICALP~'87}, volume 267 of \emph{Lect. Notes in Comput.
  Sci.}, pages 499--508. Springer, 1987.
\newblock \doi{10.1007/3-540-18088-5_43}.

\bibitem[Finkel and Schnoebelen(2001)]{finkel98b}
A.~Finkel and {\relax Ph}.~Schnoebelen.
\newblock Well-structured transition systems everywhere!
\newblock \emph{Theor. Comput. Sci.}, 256\penalty0 (1--2):\penalty0 63--92,
  2001.
\newblock \doi{10.1016/S0304-3975(00)00102-X}.

\bibitem[Fischer et~al.(1968)Fischer, Meyer, and Rosenberg]{fischer68}
P.~C. Fischer, A.~R. Meyer, and A.~L. Rosenberg.
\newblock Counter machines and counter languages.
\newblock \emph{Math. Sys. Theory}, 2\penalty0 (3):\penalty0 265--283, 1968.
\newblock \doi{10.1007/BF01694011}.

\bibitem[Friedman(1999)]{friedman99}
H.~M. Friedman.
\newblock Some decision problems of enormous complexity.
\newblock In \emph{LICS~'99}, pages 2--13. IEEE, 1999.
\newblock \doi{10.1109/LICS.1999.782577}.

\bibitem[Grzegorczyk(1953)]{grzegorczyk53}
A.~Grzegorczyk.
\newblock Some classes of recursive functions.
\newblock \emph{Rozprawy Matematyczne}, 4, 1953.
\newblock URL \url{http://matwbn.icm.edu.pl/ksiazki/rm/rm04/rm0401.pdf}.

\bibitem[Haase et~al.(2014)Haase, Schmitz, and Schnoebelen]{HaaseSS13}
C.~Haase, S.~Schmitz, and {\relax Ph}.~Schnoebelen.
\newblock The power of priority channel systems.
\newblock \emph{Logic. Meth. in Comput. Sci.}, 10\penalty0 (4:4), 2014.
\newblock \doi{10.2168/LMCS-10(4:4)2014}.

\bibitem[Hack(1976)]{hack76}
M.~Hack.
\newblock The equality problem for vector addition systems is undecidable.
\newblock \emph{Theor. Comput. Sci.}, 2\penalty0 (1):\penalty0 77--95, 1976.
\newblock \doi{10.1016/0304-3975(76)90008-6}.

\bibitem[Haddad et~al.(2012)Haddad, Schmitz, and Schnoebelen]{HSS-lics2012}
S.~Haddad, S.~Schmitz, and {\relax Ph}.~Schnoebelen.
\newblock The ordinal-recursive complexity of timed-arc {P}etri nets, data
  nets, and other enriched nets.
\newblock In \emph{LICS 2012}, pages 355--364. IEEE, 2012.
\newblock \doi{10.1109/LICS.2012.46}.

\bibitem[Hague(2014)]{hague14}
M.~Hague.
\newblock Senescent ground tree rewriting systems.
\newblock In \emph{CSL-LICS 2014}, pages 48:1--48:10. ACM, 2014.
\newblock \doi{10.1145/2603088.2603112}.

\bibitem[Halpern and Shoham(1991)]{halpern91interval}
J.~Y. Halpern and Y.~Shoham.
\newblock A propositional modal logic of time intervals.
\newblock \emph{J.~ACM}, 38\penalty0 (4):\penalty0 935--962, 1991.
\newblock \doi{10.1145/115234.115351}.

\bibitem[Hofman and Totzke(2014)]{hofman14}
P.~Hofman and P.~Totzke.
\newblock Trace inclusion for one-counter nets revisited.
\newblock In \emph{RP 2014}, volume 8762 of \emph{Lect. Notes in Comput. Sci.},
  pages 151--162. Springer, 2014.
\newblock \doi{10.1007/978-3-319-11439-2_12}.

\bibitem[Jan\v{c}ar(1995)]{jancar95b}
P.~Jan\v{c}ar.
\newblock Undecidability of bisimilarity for {Petri} nets and some related
  problems.
\newblock \emph{Theor. Comput. Sci.}, 148\penalty0 (2):\penalty0 281--301,
  1995.
\newblock \doi{10.1016/0304-3975(95)00037-W}.

\bibitem[Jan\v{c}ar(2001)]{jancar}
P.~Jan\v{c}ar.
\newblock Nonprimitive recursive complexity and undecidability for {P}etri net
  equivalences.
\newblock \emph{Theor. Comput. Sci.}, 256\penalty0 (1--2):\penalty0 23--30,
  2001.
\newblock \doi{10.1016/S0304-3975(00)00100-6}.

\bibitem[Jurdzi\'{n}ski and Lazi\'{c}(2011)]{jurdzinski2007}
M.~Jurdzi\'{n}ski and R.~Lazi\'{c}.
\newblock Alternating automata on data trees and {XP}ath satisfiability.
\newblock \emph{ACM Trans. Comput. Logic}, 12\penalty0 (3):\penalty0
  19:1--19:21, 2011.
\newblock \doi{10.1145/1929954.1929956}.

\bibitem[Karandikar and Schmitz(2013)]{fossacs/KarandikarS13}
P.~Karandikar and S.~Schmitz.
\newblock The parametric ordinal-recursive complexity of {P}ost embedding
  problems.
\newblock In \emph{FoSSaCS 2013}, volume 7794 of \emph{Lect. Notes in Comput.
  Sci.}, pages 273--288. Springer, 2013.
\newblock \doi{10.1007/978-3-642-37075-5_18}.

\bibitem[Karandikar and Schnoebelen(2012)]{KS-csr12}
P.~Karandikar and {\relax Ph}.~Schnoebelen.
\newblock Cutting through regular {P}ost embedding problems.
\newblock In \emph{CSR 2012}, volume 7353 of \emph{Lect. Notes in Comput.
  Sci.}, pages 229--240. Springer, 2012.
\newblock \doi{10.1007/978-3-642-30642-6_22}.

\bibitem[Kosaraju(1982)]{kosa}
S.~R. Kosaraju.
\newblock Decidability of reachability in vector addition systems.
\newblock In \emph{STOC~'82}, pages 267--281. ACM, 1982.
\newblock \doi{10.1145/800070.802201}.

\bibitem[Koymans(1990)]{koymans90}
R.~Koymans.
\newblock Specifying real-time properties with metric temporal logic.
\newblock \emph{Real-Time Systems}, 2\penalty0 (4):\penalty0 255--299, 1990.
\newblock \doi{10.1007/BF01995674}.

\bibitem[Kruskal(1972)]{kruskal72}
J.~B. Kruskal.
\newblock The theory of well-quasi-ordering: A frequently discovered concept.
\newblock \emph{J.~Comb. Theory A}, 13\penalty0 (3):\penalty0 297--305, 1972.
\newblock \doi{10.1016/0097-3165(72)90063-5}.

\bibitem[Lambert(1992)]{lambert}
J.-L. Lambert.
\newblock A structure to decide reachability in {P}etri nets.
\newblock \emph{Theor. Comput. Sci.}, 99\penalty0 (1):\penalty0 79--104, 1992.
\newblock ISSN 0304-3975.
\newblock \doi{10.1016/0304-3975(92)90173-D}.

\bibitem[Lasota and Walukiewicz(2008)]{ata}
S.~Lasota and I.~Walukiewicz.
\newblock Alternating timed automata.
\newblock \emph{ACM Trans. Comput. Logic}, 9\penalty0 (2):\penalty0
  10:1--10:27, 2008.
\newblock \doi{10.1145/1342991.1342994}.

\bibitem[Lazi\'c and Schmitz(2015)]{lazic14}
R.~Lazi\'c and S.~Schmitz.
\newblock Non-elementary complexities for branching {VASS}, {MELL}, and
  extensions.
\newblock \emph{ACM Trans. Comput. Logic}, 16\penalty0 (3:20):\penalty0 1--30,
  2015.
\newblock \doi{10.1145/2733375}.

\bibitem[Lazi\'c et~al.(2008)Lazi\'c, Newcomb, Ouaknine, Roscoe, and
  Worrell]{datanets}
R.~Lazi\'c, T.~Newcomb, J.~O. Ouaknine, A.~W. Roscoe, and J.~B. Worrell.
\newblock Nets with tokens which carry data.
\newblock \emph{Fund. Inform.}, 88\penalty0 (3):\penalty0 251--274, 2008.

\bibitem[Lazi\'c et~al.(2013)Lazi\'c, Ouaknine, and Worrell]{lazic13}
R.~Lazi\'c, J.~O. Ouaknine, and J.~B. Worrell.
\newblock Zeno, {H}ercules and the {H}ydra: {D}ownward rational termination is
  {A}ckermannian.
\newblock In \emph{MFCS~2013}, volume 8087 of \emph{Lect. Notes in Comput.
  Sci.}, pages 643--654. Springer, 2013.
\newblock \doi{10.1007/978-3-642-40313-2_57}.

\bibitem[Leroux(2011)]{leroux-popl2011}
J.~Leroux.
\newblock Vector addition system reachability problem\string: a short
  self-contained proof.
\newblock In \emph{POPL 2011}, pages 307--316. ACM, 2011.
\newblock \doi{10.1145/1926385.1926421}.

\bibitem[Leroux and Schmitz(2015)]{leroux15}
J.~Leroux and S.~Schmitz.
\newblock Demystifying reachability in vector addition systems.
\newblock In \emph{LICS~2015}, pages 56--67. IEEE, July 2015.
\newblock \doi{10.1109/LICS.2015.16}.

\bibitem[Lipton(1976)]{lipton76}
R.~J. Lipton.
\newblock The reachability problem requires exponential space.
\newblock Technical Report~62, Department of Computer Science, Yale University,
  Jan. 1976.
\newblock URL \url{http://www.cs.yale.edu/publications/techreports/tr63.pdf}.

\bibitem[L\"ob and Wainer(1970)]{lob70}
M.~H. L\"ob and S.~S. Wainer.
\newblock Hierarchies of number theoretic functions, {I}.
\newblock \emph{Arch. Math. Log.}, 13:\penalty0 39--51, 1970.
\newblock \doi{10.1007/BF01967649}.

\bibitem[Mayr(1981)]{mayr}
E.~W. Mayr.
\newblock An algorithm for the general {P}etri net reachability problem.
\newblock In \emph{STOC~'81}, pages 238--246. ACM, 1981.
\newblock \doi{10.1145/800076.802477}.

\bibitem[Mayr and Meyer(1981)]{fct}
E.~W. Mayr and A.~R. Meyer.
\newblock The complexity of the finite containment problem for {P}etri nets.
\newblock \emph{J.~ACM}, 28\penalty0 (3):\penalty0 561--576, 1981.
\newblock \doi{10.1145/322261.322271}.

\bibitem[McAloon(1984)]{mcaloon}
K.~McAloon.
\newblock Petri nets and large finite sets.
\newblock \emph{Theor. Comput. Sci.}, 32\penalty0 (1--2):\penalty0 173--183,
  1984.
\newblock \doi{10.1016/0304-3975(84)90029-X}.

\bibitem[Meyer(1975{\natexlab{a}})]{meyer72}
A.~R. Meyer.
\newblock Weak monadic second order theory of successor is not
  elementary-recursive.
\newblock In \emph{Logic Colloquium '72--73}, volume 453 of \emph{Lect. Notes
  Math.}, pages 132--154. Springer, 1975{\natexlab{a}}.
\newblock \doi{10.1007/BFb0064872}.

\bibitem[Meyer(1975{\natexlab{b}})]{meyer74}
A.~R. Meyer.
\newblock The inherent computational complexity of theories of ordered sets.
\newblock In \emph{ICM~'74 Vol.~2}, pages 477--482. Canadian Mathematical
  Congress, 1975{\natexlab{b}}.
\newblock URL
  \url{http://www.mathunion.org/ICM/ICM1974.2/Main/icm1974.2.0477.0482.ocr.pdf}.

\bibitem[Meyer and Ritchie(1967)]{meyer67}
A.~R. Meyer and D.~M. Ritchie.
\newblock The complexity of loop programs.
\newblock In \emph{ACM '67}, pages 465--469, 1967.
\newblock \doi{10.1145/800196.806014}.

\bibitem[Montanari et~al.(2010)Montanari, Puppis, and Sala]{montanari10}
A.~Montanari, G.~Puppis, and P.~Sala.
\newblock Maximal decidable fragments of {H}alpern and {S}hoham’s modal logic
  of intervals.
\newblock In \emph{ICALP~2010}, volume 6199 of \emph{Lect. Notes in Comput.
  Sci.}, pages 345--356. Springer, 2010.
\newblock \doi{10.1007/978-3-642-14162-1_29}.

\bibitem[Odifreddi(1999)]{odifreddi99}
P.~Odifreddi.
\newblock \emph{Classical Recursion Theory, vol.\ II}, volume 143 of
  \emph{Studies in Logic and the Foundations of Mathematics}.
\newblock Elsevier, 1999.
\newblock \doi{10.1016/S0049-237X(99)80040-8}.

\bibitem[Omri and Weiermann(2009)]{omri09}
E.~Omri and A.~Weiermann.
\newblock Classifying the phase transition threshold for {A}ckermannian
  functions.
\newblock \emph{Ann. Pure Appl. Log.}, 158\penalty0 (3):\penalty0 156--162,
  2009.
\newblock \doi{10.1016/j.apal.2007.02.004}.

\bibitem[Ouaknine and Worrell(2006)]{ouaknine06}
J.~O. Ouaknine and J.~B. Worrell.
\newblock On {M}etric {T}emporal {L}ogic and faulty {T}uring machines.
\newblock In \emph{FoSSaCS~2006}, volume 3921 of \emph{Lect. Notes in Comput.
  Sci.}, pages 217--230. Springer, 2006.
\newblock \doi{10.1007/11690634_15}.

\bibitem[Ouaknine and Worrell(2007)]{mtl}
J.~O. Ouaknine and J.~B. Worrell.
\newblock On the decidability and complexity of {M}etric {T}emporal {L}ogic
  over finite words.
\newblock \emph{Logic. Meth. in Comput. Sci.}, 3\penalty0 (1:8), 2007.
\newblock \doi{10.2168/LMCS-3(1:8)2007}.

\bibitem[Rackoff(1978)]{rackoff78}
C.~Rackoff.
\newblock The covering and boundedness problems for vector addition systems.
\newblock \emph{Theor. Comput. Sci.}, 6\penalty0 (2):\penalty0 223--231, 1978.
\newblock \doi{10.1016/0304-3975(78)90036-1}.

\bibitem[Ritchie(1963)]{ritchie63}
R.~W. Ritchie.
\newblock Classes of predictably computable functions.
\newblock \emph{Trans. Amer. Math. Soc.}, 106\penalty0 (1):\penalty0 139--173,
  1963.
\newblock \doi{10.1090/S0002-9947-1963-0158822-2}.

\bibitem[Ritchie(1965)]{ritchie65}
R.~W. Ritchie.
\newblock Classes of recursive functions based on {A}ckermann's function.
\newblock \emph{Pac. J.~Math.}, 15\penalty0 (3):\penalty0 1027--1044, 1965.
\newblock \doi{10.2140/pjm.1965.15.1027}.

\bibitem[Rosa-Velardo(2014)]{rosavelardo14}
F.~Rosa-Velardo.
\newblock Ordinal recursive complexity of unordered data nets.
\newblock Technical Report TR-4-14, Departamento de Sistemas Inform\'aticos y
  Computaci\'on, Universidad Complutense de Madrid, 2014.
\newblock URL
  \url{https://federwin.sip.ucm.es/sic/investigacion/publicaciones/pdfs/TR-04-14.pdf}.

\bibitem[Rose(1984)]{rose84}
H.~E. Rose.
\newblock \emph{Subrecursion: Functions and Hierarchies}, volume~9 of
  \emph{Oxford Logic Guides}.
\newblock Clarendon Press, 1984.

\bibitem[Schmitz and Schnoebelen(2011)]{SS-icalp2011}
S.~Schmitz and {\relax Ph}.~Schnoebelen.
\newblock Multiply-recursive upper bounds with {Higman}'s {L}emma.
\newblock In \emph{ICALP~2011}, volume 6756 of \emph{Lect. Notes in Comput.
  Sci.}, pages 441--452. Springer, 2011.
\newblock \doi{10.1007/978-3-642-22012-8_35}.

\bibitem[Schmitz and Schnoebelen(2012)]{wqo}
S.~Schmitz and {\relax Ph}.~Schnoebelen.
\newblock Algorithmic aspects of {WQO} theory.
\newblock Lecture notes, 2012.
\newblock URL \url{http://cel.archives-ouvertes.fr/cel-00727025}.

\bibitem[Schmitz and Schnoebelen(2013)]{concur/SchmitzS13}
S.~Schmitz and {\relax Ph}.~Schnoebelen.
\newblock The power of well-structured systems.
\newblock In \emph{Concur 2013}, volume 8052 of \emph{Lect. Notes in Comput.
  Sci.}, pages 5--24. Springer, 2013.
\newblock \doi{10.1007/978-3-642-40184-8_2}.

\bibitem[Schnoebelen(2002)]{phs-IPL2002}
{\relax Ph}.~Schnoebelen.
\newblock Verifying lossy channel systems has nonprimitive recursive
  complexity.
\newblock \emph{Inf. Process. Lett.}, 83\penalty0 (5):\penalty0 251--261, 2002.
\newblock \doi{10.1016/S0020-0190(01)00337-4}.

\bibitem[Schnoebelen(2010)]{phs-mfcs2010}
{\relax Ph}.~Schnoebelen.
\newblock Revisiting {A}ckermann-hardness for lossy counter machines and reset
  {P}etri nets.
\newblock In \emph{MFCS 2010}, volume 6281 of \emph{Lect. Notes in Comput.
  Sci.}, pages 616--628. Springer, 2010.
\newblock \doi{10.1007/978-3-642-15155-2_54}.

\bibitem[Schwichtenberg(1982)]{schwichtenberg82}
H.~Schwichtenberg.
\newblock Complexity of normalization in the pure typed lambda-calculus.
\newblock In \emph{L.E.J. Brouwer Centenary Symposium}, volume 110 of
  \emph{Studies in Logic and the Foundations of Mathematics}, pages 453--457.
  Elsevier, 1982.
\newblock \doi{10.1016/S0049-237X(09)70143-0}.

\bibitem[Schwichtenberg and Wainer(2012)]{sw12}
H.~Schwichtenberg and S.~S. Wainer.
\newblock \emph{Proofs and Computation}.
\newblock Perspectives in Logic. Cambridge University Press, 2012.

\bibitem[Statman(1979)]{statman79}
R.~Statman.
\newblock The typed $\lambda$-calculus is not elementary recursive.
\newblock \emph{Theor. Comput. Sci.}, 9\penalty0 (1):\penalty0 73--81, 1979.
\newblock \doi{10.1016/0304-3975(79)90007-0}.

\bibitem[Stockmeyer and Meyer(1973)]{stockmeyer73}
L.~J. Stockmeyer and A.~R. Meyer.
\newblock Word problems requiring exponential time.
\newblock In \emph{STOC~'73}, pages 1--9. ACM, 1973.
\newblock \doi{10.1145/800125.804029}.

\bibitem[Tan(2010)]{tan2010}
T.~Tan.
\newblock On pebble automata for data languages with decidable emptiness
  problem.
\newblock \emph{J.~Comput. Syst. Sci.}, 76\penalty0 (8):\penalty0 778--791,
  2010.
\newblock \doi{10.1016/j.jcss.2010.03.004}.

\bibitem[Tzevelekos and Grigore(2013)]{tzelevekos13}
N.~Tzevelekos and R.~Grigore.
\newblock History-register automata.
\newblock In \emph{FoSSaCS 2013}, volume 7794 of \emph{Lect. Notes in Comput.
  Sci.}, pages 273--288, 2013.
\newblock \doi{10.1007/978-3-642-37075-5_2}.

\bibitem[Urquhart(1984)]{urquhart84}
A.~Urquhart.
\newblock The undecidability of entailment and relevant implication.
\newblock \emph{J.~Symb. Log.}, 49\penalty0 (4):\penalty0 1059--1073, 1984.
\newblock \doi{10.2307/2274261}.

\bibitem[Urquhart(1999)]{urquhart99}
A.~Urquhart.
\newblock The complexity of decision procedures in relevance logic {II}.
\newblock \emph{J.~Symb. Log.}, 64\penalty0 (4):\penalty0 1774--1802, 1999.
\newblock \doi{10.2307/2586811}.

\bibitem[Vorobyov(2004)]{vorobyov04}
S.~Vorobyov.
\newblock The most nonelementary theory.
\newblock \emph{Inform. and Comput.}, 190\penalty0 (2):\penalty0 196--219,
  2004.
\newblock \doi{10.1016/j.ic.2004.02.002}.

\bibitem[Wainer(1970)]{wainer70}
S.~S. Wainer.
\newblock A classification of the ordinal recursive functions.
\newblock \emph{Arch. Math. Log.}, 13\penalty0 (3):\penalty0 136--153, 1970.
\newblock \doi{10.1007/BF01973619}.

\bibitem[Weiermann(1994)]{weiermann94}
A.~Weiermann.
\newblock Complexity bounds for some finite forms of {K}ruskal's {T}heorem.
\newblock \emph{J.~Symb. Comput.}, 18\penalty0 (5):\penalty0 463--488, 1994.
\newblock \doi{10.1006/jsco.1994.1059}.

\end{thebibliography}


\begin{thebibliography}{0}
\providecommand{\natexlab}[1]{#1}
\providecommand{\url}[1]{\texttt{#1}}
\expandafter\ifx\csname urlstyle\endcsname\relax
  \providecommand{\doi}[1]{doi: #1}\else
  \providecommand{\doi}{doi: \begingroup \urlstyle{rm}\Url}\fi

\end{thebibliography}
